%% file: Order_Statistics_arXiv.tex
\documentclass[preprint,authoryear,12pt]{elsarticle}

\usepackage{floatrow}
\usepackage{amsmath}
\usepackage{amssymb}

\journal{}

\usepackage{amsmath}
\usepackage{amsthm}
\newtheorem{definition}{Definition}

\newtheorem{lemma}{Lemma}
\newtheorem{remark}{Remark}

\usepackage{optidef}
\usepackage{subcaption}

\allowdisplaybreaks

\newcommand\numberthis{\addtocounter{equation}{1}\tag{\theequation}}
\usepackage{todonotes}

\usepackage{calc}
\usepackage{algorithm}
\usepackage{algpseudocode}
\usepackage{enumitem}
\usepackage{url}

\usepackage{hyperref}
\usepackage{breakurl}

\newcommand{\imax}{\vee}

\makeatletter
\newenvironment{proofof}[1]{\par
	\pushQED{\qed}%
	\normalfont \topsep6\p@\@plus6\p@\relax
	\trivlist
	\item[\hskip\labelsep
	\bfseries
	Proof of #1\@addpunct{.}]\ignorespaces
}{%
	\popQED\endtrivlist\@endpefalse
}
\makeatother

\usepackage{mathtools}

\mathtoolsset{showonlyrefs}

\begin{document}

\begin{frontmatter}



\title{Efficient Calculation of the Joint Distribution of Order Statistics}

\author{Jonathan von Schroeder}
\author{Thorsten Dickhaus\corref{cor1}}
\cortext[cor1]{Institute for Statistics, University of Bremen, 
P. O. Box 330 440, 28344 Bremen, Germany. Tel: +49 421 218-63651. E-mail address: dickhaus@uni-bremen.de (Thorsten Dickhaus).}

\address{Institute for Statistics, University of Bremen, Germany}

\begin{abstract}
We consider the problem of computing the joint distribution of order statistics of stochastically independent random variables in one- and two-group models. While recursive formulas for evaluating the joint cumulative distribution function of such order statistics exist in the literature for a longer time, their numerical implementation remains a challenging task. We tackle this task by presenting novel generalizations of known recursions which we utilize to obtain exact results (calculated in rational arithmetic) as well as faithfully rounded results. Finally, some applications in stepwise multiple hypothesis testing are discussed.
\end{abstract}

\end{frontmatter}


\input{./1-intro}

\input{./2-order-statistics}
\input{./3-generalized-recursions}
\input{./4-exact-bolshev}
\input{./5-faithful-noe}
\input{./6-applications}
\input{./7-discussion}

\section*{Acknowledgments}
Jonathan von Schroeder is supported by the Deutsche Forschungsgemeinschaft (DFG) within the framework of RTG 2224  
"$\pi^3$: Parameter Identification - Analysis, Algorithms, Applications".

\appendix
\include{./a-proofs}



\bibliographystyle{elsarticle-harv}
\bibliography{bibliography}

\end{document}

%% file: 1-intro.tex
\section{Introduction}
The joint distribution of order statistics $X_{1:n}, \hdots, X_{n:n}$ of stochastically independent random variables $X_1, \hdots, X_n$ plays a pivotal role in the theory of empirical processes and in
nonparametric statistics; see, e.\ g., \cite{shorackwellner} and
\cite{Dickhaus-nonparametric}. For instance, the exact finite-sample null distributions of classical goodness-of-fit tests like the Kolmogorov-Smirnov and the Cram\'{e}r-von Mises test as well as those of modern "higher criticism" goodness-of-fit tests rely on such joint distributions; cf.\ \cite{FiGo-BiomJ}, \cite{FiGo-Bernoulli}, and \cite{FiGo-Annals} for recent developments and further references. In simultaneous statistical inference, the joint distribution of ordered $p$-values is needed to analyze the type I and type II error behaviour of stepwise rejective multiple test procedures; cf.\ Chapter 5 of \cite{Dickhaus-Buch2014}. 

In the case that $X_1, \hdots, X_n$ are identically distributed (we refer to this case as a one-group model), classical recursive methods like Bolshev's recursion, Noe's recursion, and Steck's recursion allow for computing the joint cumulative distribution function (cdf) of $X_{1:n}, \hdots, X_{n:n}$ exactly; cf.\ Section 9.3 of \cite{shorackwellner}. A generalization of Steck's recursion to two-group models has been introduced by \cite{blanchard2014least}. The other aforementioned recursions can be generalized in an analogous manner, as we will demonstrate in Section \ref{sec3} of the present work.

While conceptually appealing, numerical properties of the aforementioned recursions are not well understood yet, and existing implementations into computer software often refer to rule-of-thumb-type upper bounds on $n$ such that the respective implementation is trustworthy. For example, Art B. Owen reports in his implementation of the two-sided version of Noe's recursion in C (see \url{https://www.stat.washington.edu/jaw/RESEARCH/SOFTWARE/BERKJONES/BJ-RBJ-C-Code/noe.c}) that the recursion works well for $n\leq 1000$ but "For larger n (eg 1800 or more) [...] unexplained odd behavior." Similarly, in the R Package mutoss (cf.\ \cite{mutoss}) the following comment is made on the implementation of Bolshev's recursion: "Because of numerical issues $n$ should not be greater than 100." Recently, \cite{moscovich2017} introduced a computational method for one-group models. However, they do not consider the numerical accuracy of their approach rigorously.

In this work, we contribute to the analysis of the numerical accuracy and the computational complexity of existing approaches for computing the joint distribution of $X_{1:n}, \hdots, X_{n:n}$ in a mathematically rigorous manner. Furthermore, we provide novel computational techniques for one- and two-group models which are guaranteed to provide accurate results for arbitrary sample size $n$. The rest of the material is structured as follows. In Section \ref{sec2}, we introduce the relevant quantities. The (generalized) recursions for one- and two-group models are provided in Section \ref{sec3}, together with a rigorous analysis of their computational complexities and their numerical properties. Our proposed exact computational methods rely on rational arithmetic (Section \ref{sec:exact_bolshev}) and on faithful rounding (Section \ref{sec:noe_faithful}), respectively. Applications in multiple hypothesis testing are given in Section \ref{sec6}, and we conclude with a discussion in Section \ref{sec7}. Lengthy proofs as well as pseudo code for the considered algorithms are deferred to the Appendix.

%% file: 2-order-statistics.tex
\section{Order Statistics} \label{sec2}
Throughout the following sections, we let $[n]:= \{1, 2, \hdots, n\}$ for a natural number $n\in\mathbb{N}$.
Consider stochastically independent, real-valued random variables $X_1,\hdots,X_n$, which are all driven by the same probability measure $\mathbb{P}$. Let $I_1:=[n]$, and recursively define 
\begin{align*}
    i_j&:=\arg\min_{i\in I_j} X_i,\\
    I_{j}&:=I_{j-1}\setminus\left\{i_{j-1}\right\}
\end{align*} for $j\in[n]$. Then we call $X_{i_1},\hdots,X_{i_n}$ the order statistics of $X_1,\hdots,X_n$,  which we will denote by $X_{1:n},\hdots,X_{n:n}$ in the remainder. The random variable $X_{i:n}$ will be called the $i$-th order statistic of the random vector $(X_1,\hdots,X_n)^\top$. 
Let $F_i$ denote the marginal cdf of $X_i$ for $i \in [n]$. 
This paper will present methods for the quick and numerically stable calculation of 
\begin{align*}
 \Psi_{n_1,n_2}^{G_1,G_2}(\mathbf{b}):=\mathbb{P}\left(X_{1:n}\leq b_1,\hdots,X_{n:n}\leq b_n\right), \; \mathbf{b} = (b_1, \hdots, b_n)^\top \in \mathbb{R}^n,
\end{align*}
assuming that $\forall i\in[n]:F_i\in\{G_1,G_2\}$ where $G_1,G_2$ are two continuous distribution functions on $\mathbb{R}$ and with $n_i =\left|\left\{j\in[n]: F_j=G_i\right\}\right|$ denoting 
the number of $X_j$'s distributed according to $G_i$, $i=1,2$.
Since it holds that
\begin{align*}
    G_1(X_i) \sim
		\begin{cases}
        \text{Uni}{[0,1]}, &F_i = G_1,\\
        G_2\circ G_1^{-1}, &F_i = G_2,
    \end{cases}
\end{align*} 
it follows that $\Psi_{n_1,n_2}^{G_1,G_2}=\Psi_{n_1,n_2}^{\text{Uni}{[0,1]},F}\circ G_1$, where $F:=G_2\circ G_1^{-1}$. Therefore, it is sufficient to consider the calculation of 
$\Psi_{n_1,n_2}^{\text{Uni}{[0,1]},F}(\mathbf{b})$ 
for an arbitrary continuous distribution function $F:[0,1]\rightarrow[0,1]$ and 
argument $\mathbf{b} \in[0,1]^n$. In the sequel, we suppress the dependence on $F$ and $\mathbf{b}$ notationally, and write $\Psi(n_1,n_2):= \Psi_{n_1,n_2}^{\text{Uni}{[0,1]},F}(\mathbf{b})$ for notational convenience.

As outlined in the Introduction, for $n_2=0$ there exist many well known recursions (see e.\ g. Section 9.3 of \cite{shorackwellner}) for computing $\Psi(n_1,n_2)$. There are also newer approaches based on numerical integration(see \cite{moscovich2016}) or based on the Poisson process (see \cite{moscovich2017}). Unfortunately, the former cannot be easily generalized to the case $0 < n_2 < n$, since the Lebesgue density of an order statistic is in general not piece-wise constant. The latter is very fast due to usage of the Fourier transform, but numerically unstable for small values of the $b_i$'s. This can for instance be  demonstrated using the thresholds of the well-known linear step-up test (cf. \cite{benjamini1995}) for control of the false discovery rate (FDR); see Figure \ref{fig:relative_error}.
\cite{glueck2008fast} proposed an algorithm with exponential $O(n^n)$ complexity (cf. \cite[Theorem 4.2]{glueck2008fast}), resulting in a very high computational effort for moderate or large values of $n$. However, the method of \cite{glueck2008fast} can be used to compute $k$-variate marginal distributions for $k < < n$, because in such cases the complexity of their approach reduces to $O(n^k)$.

Since we are mostly concerned with the full joint distribution, we  extend the approach suggested by \cite{blanchard2014least} and provide generalizations of Bolshev's and Noe's recursions. We compare them to the generalization of Steck's recursion proposed by \cite{blanchard2014least} and demonstrate that the Bolshev recursion is suitable for exact computations in rational arithmetic, whereas Noe's recursion is numerically stable when computed in fixed-precision floating point arithmetic.

All our numerical calculations were performed on a Windows 7 machine with an Intel(R) Core(TM) i7-4790 CPU with 32 gigabytes of RAM.

\begin{figure}[!htb]
	\centering
    \input{./relative_numerical_error}
	\caption{Relative error (on the $\log_{10}$ scale) of the methods presented by \cite{moscovich2016} and \cite{moscovich2017}, respectively, when calculating $\Psi(n, 0)$ for the thresholds $b_i \equiv b_i^{(n)}:=0.05 \times i / n$. A value of $-16$ implies at least 15 accurate non-zero digits in base 10. For $n\leq 77$ the relative error of the three methods is visually barely distinguishable. For $n\geq 77$ the dotted line below zero corresponds to the "Numerical Integration".}
	\label{fig:relative_error}
\end{figure}
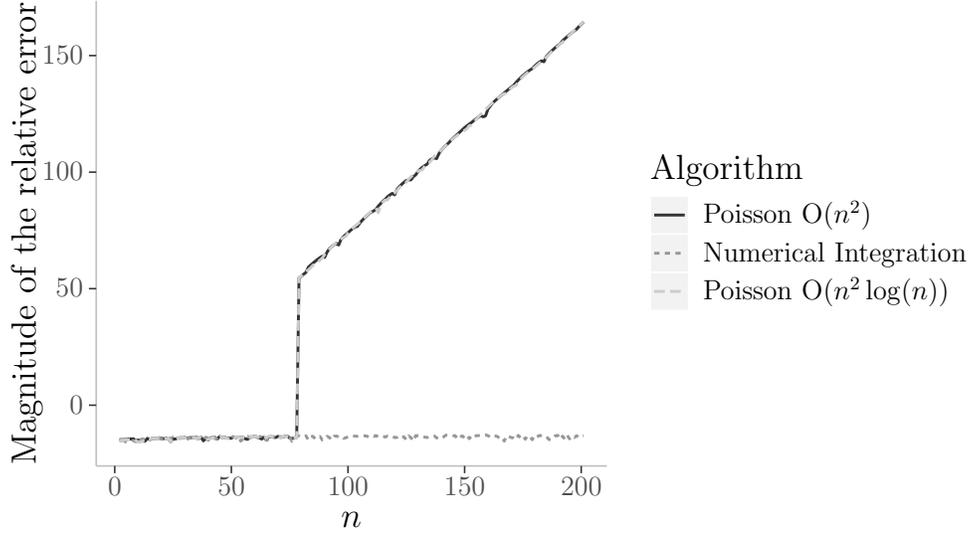

%% file: relative_numerical_error.tex
\begin{tikzpicture}[x=1pt,y=1pt]
\definecolor{fillColor}{RGB}{255,255,255}
\begin{scope}
\definecolor{drawColor}{RGB}{255,255,255}
\definecolor{fillColor}{RGB}{255,255,255}

\path[draw=drawColor,line width= 0.6pt,line join=round,line cap=round,fill=fillColor] (  0.00, 74.53) rectangle (361.35,286.82);
\end{scope}
\begin{scope}
\definecolor{drawColor}{RGB}{255,255,255}

\path[draw=drawColor,line width= 0.3pt,line join=round] ( 35.92,106.20) --
	(229.05,106.20);

\path[draw=drawColor,line width= 0.3pt,line join=round] ( 35.92,150.31) --
	(229.05,150.31);

\path[draw=drawColor,line width= 0.3pt,line join=round] ( 35.92,194.43) --
	(229.05,194.43);

\path[draw=drawColor,line width= 0.3pt,line join=round] ( 35.92,238.54) --
	(229.05,238.54);

\path[draw=drawColor,line width= 0.3pt,line join=round] ( 64.99,105.26) --
	( 64.99,281.32);

\path[draw=drawColor,line width= 0.3pt,line join=round] (109.10,105.26) --
	(109.10,281.32);

\path[draw=drawColor,line width= 0.3pt,line join=round] (153.21,105.26) --
	(153.21,281.32);

\path[draw=drawColor,line width= 0.3pt,line join=round] (197.33,105.26) --
	(197.33,281.32);

\path[draw=drawColor,line width= 0.6pt,line join=round] ( 35.92,128.26) --
	(229.05,128.26);

\path[draw=drawColor,line width= 0.6pt,line join=round] ( 35.92,172.37) --
	(229.05,172.37);

\path[draw=drawColor,line width= 0.6pt,line join=round] ( 35.92,216.48) --
	(229.05,216.48);

\path[draw=drawColor,line width= 0.6pt,line join=round] ( 35.92,260.60) --
	(229.05,260.60);

\path[draw=drawColor,line width= 0.6pt,line join=round] ( 42.93,105.26) --
	( 42.93,281.32);

\path[draw=drawColor,line width= 0.6pt,line join=round] ( 87.04,105.26) --
	( 87.04,281.32);

\path[draw=drawColor,line width= 0.6pt,line join=round] (131.16,105.26) --
	(131.16,281.32);

\path[draw=drawColor,line width= 0.6pt,line join=round] (175.27,105.26) --
	(175.27,281.32);

\path[draw=drawColor,line width= 0.6pt,line join=round] (219.38,105.26) --
	(219.38,281.32);
\definecolor{drawColor}{gray}{0.20}

\path[draw=drawColor,line width= 1.1pt,line join=round] ( 44.70,114.81) --
	( 45.58,114.99) --
	( 46.46,115.30) --
	( 47.34,115.30) --
	( 48.22,115.51) --
	( 49.11,115.46) --
	( 49.99,115.57) --
	( 50.87,114.83) --
	( 51.75,115.06) --
	( 52.64,114.80) --
	( 53.52,114.65) --
	( 54.40,114.61) --
	( 55.28,115.72) --
	( 56.16,115.25) --
	( 57.05,115.62) --
	( 57.93,115.58) --
	( 58.81,115.67) --
	( 59.69,115.73) --
	( 60.58,115.73) --
	( 61.46,115.56) --
	( 62.34,115.71) --
	( 63.22,115.52) --
	( 64.11,114.43) --
	( 64.99,115.52) --
	( 65.87,115.45) --
	( 66.75,115.77) --
	( 67.63,115.97) --
	( 68.52,115.92) --
	( 69.40,116.04) --
	( 70.28,115.75) --
	( 71.16,115.93) --
	( 72.05,116.00) --
	( 72.93,116.10) --
	( 73.81,115.77) --
	( 74.69,115.62) --
	( 75.57,116.22) --
	( 76.46,114.74) --
	( 77.34,115.65) --
	( 78.22,115.08) --
	( 79.10,115.85) --
	( 79.99,116.16) --
	( 80.87,115.73) --
	( 81.75,115.46) --
	( 82.63,116.45) --
	( 83.52,115.99) --
	( 84.40,116.14) --
	( 85.28,115.82) --
	( 86.16,115.81) --
	( 87.04,115.72) --
	( 87.93,116.32) --
	( 88.81,115.96) --
	( 89.69,116.29) --
	( 90.57,116.07) --
	( 91.46,116.04) --
	( 92.34,115.51) --
	( 93.22,115.31) --
	( 94.10,116.34) --
	( 94.98,116.23) --
	( 95.87,116.58) --
	( 96.75,116.24) --
	( 97.63,115.89) --
	( 98.51,115.42) --
	( 99.40,116.31) --
	(100.28,115.17) --
	(101.16,116.30) --
	(102.04,116.32) --
	(102.92,116.11) --
	(103.81,116.02) --
	(104.69,116.48) --
	(105.57,116.09) --
	(106.45,116.16) --
	(107.34,116.45) --
	(108.22,116.65) --
	(109.10,115.81) --
	(109.98,115.67) --
	(110.87,115.81) --
	(111.75,116.18) --
	(112.63,176.48) --
	(113.51,176.97) --
	(114.39,177.85) --
	(115.28,178.53) --
	(116.16,180.51) --
	(117.04,181.29) --
	(117.92,182.05) --
	(118.81,182.78) --
	(119.69,183.48) --
	(120.57,184.13) --
	(121.45,184.68) --
	(122.33,184.18) --
	(123.22,186.15) --
	(124.10,187.55) --
	(124.98,188.31) --
	(125.86,188.98) --
	(126.75,189.52) --
	(127.63,188.87) --
	(128.51,191.02) --
	(129.39,192.01) --
	(130.28,192.82) --
	(131.16,193.54) --
	(132.04,194.37) --
	(132.92,195.11) --
	(133.80,195.72) --
	(134.69,195.90) --
	(135.57,197.07) --
	(136.45,198.16) --
	(137.33,199.07) --
	(138.22,199.90) --
	(139.10,200.70) --
	(139.98,201.36) --
	(140.86,201.78) --
	(141.74,202.41) --
	(142.63,203.63) --
	(143.51,204.59) --
	(144.39,205.46) --
	(145.27,206.27) --
	(146.16,207.20) --
	(147.04,207.89) --
	(147.92,208.45) --
	(148.80,207.86) --
	(149.69,209.96) --
	(150.57,210.98) --
	(151.45,211.85) --
	(152.33,212.64) --
	(153.21,213.35) --
	(154.10,213.51) --
	(154.98,213.54) --
	(155.86,214.74) --
	(156.74,215.41) --
	(157.63,216.70) --
	(158.51,217.64) --
	(159.39,218.43) --
	(160.27,218.99) --
	(161.15,219.54) --
	(162.04,220.96) --
	(162.92,221.54) --
	(163.80,221.85) --
	(164.68,221.58) --
	(165.57,223.51) --
	(166.45,224.73) --
	(167.33,225.72) --
	(168.21,226.60) --
	(169.10,227.42) --
	(169.98,228.36) --
	(170.86,229.23) --
	(171.74,230.05) --
	(172.62,230.87) --
	(173.51,231.70) --
	(174.39,232.52) --
	(175.27,233.31) --
	(176.15,234.04) --
	(177.04,234.69) --
	(177.92,235.29) --
	(178.80,236.04) --
	(179.68,236.74) --
	(180.56,237.36) --
	(181.45,237.80) --
	(182.33,237.41) --
	(183.21,237.78) --
	(184.09,240.33) --
	(184.98,241.51) --
	(185.86,242.62) --
	(186.74,243.46) --
	(187.62,244.25) --
	(188.50,245.00) --
	(189.39,245.73) --
	(190.27,246.47) --
	(191.15,247.22) --
	(192.03,247.96) --
	(192.92,248.67) --
	(193.80,249.35) --
	(194.68,250.40) --
	(195.56,251.34) --
	(196.45,252.22) --
	(197.33,253.03) --
	(198.21,253.79) --
	(199.09,254.49) --
	(199.97,255.12) --
	(200.86,255.69) --
	(201.74,256.61) --
	(202.62,257.41) --
	(203.50,258.11) --
	(204.39,258.65) --
	(205.27,258.24) --
	(206.15,260.23) --
	(207.03,261.22) --
	(207.91,262.10) --
	(208.80,262.92) --
	(209.68,263.72) --
	(210.56,264.40) --
	(211.44,265.16) --
	(212.33,266.21) --
	(213.21,267.02) --
	(214.09,267.79) --
	(214.97,268.53) --
	(215.86,269.24) --
	(216.74,269.93) --
	(217.62,270.61) --
	(218.50,271.59) --
	(219.38,272.46) --
	(220.27,273.29);
\definecolor{drawColor}{RGB}{152,152,152}

\path[draw=drawColor,line width= 1.1pt,dash pattern=on 2pt off 2pt ,line join=round] ( 44.70,115.21) --
	( 45.58,115.46) --
	( 46.46,113.26) --
	( 47.34,115.60) --
	( 48.22,115.13) --
	( 49.11,115.17) --
	( 49.99,114.81) --
	( 50.87,114.41) --
	( 51.75,115.85) --
	( 52.64,115.30) --
	( 53.52,115.48) --
	( 54.40,115.29) --
	( 55.28,115.03) --
	( 56.16,114.91) --
	( 57.05,115.82) --
	( 57.93,115.37) --
	( 58.81,115.41) --
	( 59.69,116.07) --
	( 60.58,115.33) --
	( 61.46,115.56) --
	( 62.34,115.81) --
	( 63.22,115.73) --
	( 64.11,115.67) --
	( 64.99,116.16) --
	( 65.87,115.68) --
	( 66.75,115.97) --
	( 67.63,115.63) --
	( 68.52,115.79) --
	( 69.40,116.04) --
	( 70.28,115.07) --
	( 71.16,116.09) --
	( 72.05,115.05) --
	( 72.93,116.02) --
	( 73.81,116.29) --
	( 74.69,115.83) --
	( 75.57,115.81) --
	( 76.46,115.63) --
	( 77.34,116.11) --
	( 78.22,115.82) --
	( 79.10,116.41) --
	( 79.99,115.96) --
	( 80.87,116.24) --
	( 81.75,115.08) --
	( 82.63,115.50) --
	( 83.52,116.28) --
	( 84.40,115.99) --
	( 85.28,116.46) --
	( 86.16,116.31) --
	( 87.04,116.44) --
	( 87.93,116.24) --
	( 88.81,116.16) --
	( 89.69,116.58) --
	( 90.57,114.64) --
	( 91.46,115.42) --
	( 92.34,116.07) --
	( 93.22,116.50) --
	( 94.10,115.81) --
	( 94.98,116.28) --
	( 95.87,116.32) --
	( 96.75,115.93) --
	( 97.63,116.41) --
	( 98.51,116.26) --
	( 99.40,114.65) --
	(100.28,115.43) --
	(101.16,114.41) --
	(102.04,116.35) --
	(102.92,116.39) --
	(103.81,115.86) --
	(104.69,115.74) --
	(105.57,116.42) --
	(106.45,116.25) --
	(107.34,116.46) --
	(108.22,116.43) --
	(109.10,116.39) --
	(109.98,114.53) --
	(110.87,116.17) --
	(111.75,116.44) --
	(112.63,116.61) --
	(113.51,116.14) --
	(114.39,116.48) --
	(115.28,116.65) --
	(116.16,115.05) --
	(117.04,116.33) --
	(117.92,116.50) --
	(118.81,115.43) --
	(119.69,115.97) --
	(120.57,116.42) --
	(121.45,115.06) --
	(122.33,115.28) --
	(123.22,116.50) --
	(124.10,116.03) --
	(124.98,115.08) --
	(125.86,116.49) --
	(126.75,116.59) --
	(127.63,116.24) --
	(128.51,116.54) --
	(129.39,115.74) --
	(130.28,115.89) --
	(131.16,116.32) --
	(132.04,116.71) --
	(132.92,116.55) --
	(133.80,116.55) --
	(134.69,116.49) --
	(135.57,116.64) --
	(136.45,116.68) --
	(137.33,116.67) --
	(138.22,116.58) --
	(139.10,116.08) --
	(139.98,116.55) --
	(140.86,116.38) --
	(141.74,116.80) --
	(142.63,116.25) --
	(143.51,116.29) --
	(144.39,116.56) --
	(145.27,116.17) --
	(146.16,116.62) --
	(147.04,116.60) --
	(147.92,116.64) --
	(148.80,115.98) --
	(149.69,116.07) --
	(150.57,116.56) --
	(151.45,115.90) --
	(152.33,116.32) --
	(153.21,115.07) --
	(154.10,116.64) --
	(154.98,115.98) --
	(155.86,114.65) --
	(156.74,116.69) --
	(157.63,116.73) --
	(158.51,116.78) --
	(159.39,116.43) --
	(160.27,116.14) --
	(161.15,116.03) --
	(162.04,116.27) --
	(162.92,116.63) --
	(163.80,116.62) --
	(164.68,115.91) --
	(165.57,115.46) --
	(166.45,115.57) --
	(167.33,116.27) --
	(168.21,116.93) --
	(169.10,116.84) --
	(169.98,116.90) --
	(170.86,116.72) --
	(171.74,116.72) --
	(172.62,115.35) --
	(173.51,116.21) --
	(174.39,116.66) --
	(175.27,115.77) --
	(176.15,116.78) --
	(177.04,115.56) --
	(177.92,116.72) --
	(178.80,116.09) --
	(179.68,116.74) --
	(180.56,116.74) --
	(181.45,117.02) --
	(182.33,116.91) --
	(183.21,117.03) --
	(184.09,116.71) --
	(184.98,116.20) --
	(185.86,116.79) --
	(186.74,116.65) --
	(187.62,116.44) --
	(188.50,115.75) --
	(189.39,116.10) --
	(190.27,114.64) --
	(191.15,116.72) --
	(192.03,116.36) --
	(192.92,116.01) --
	(193.80,115.83) --
	(194.68,115.20) --
	(195.56,116.83) --
	(196.45,115.97) --
	(197.33,116.78) --
	(198.21,115.35) --
	(199.09,116.71) --
	(199.97,116.46) --
	(200.86,115.16) --
	(201.74,116.79) --
	(202.62,116.22) --
	(203.50,116.80) --
	(204.39,116.68) --
	(205.27,116.46) --
	(206.15,116.57) --
	(207.03,115.32) --
	(207.91,115.11) --
	(208.80,116.25) --
	(209.68,116.72) --
	(210.56,116.82) --
	(211.44,115.36) --
	(212.33,116.55) --
	(213.21,115.30) --
	(214.09,116.85) --
	(214.97,115.50) --
	(215.86,116.64) --
	(216.74,116.03) --
	(217.62,116.72) --
	(218.50,115.98) --
	(219.38,116.49) --
	(220.27,116.69);
\definecolor{drawColor}{gray}{0.80}

\path[draw=drawColor,line width= 1.1pt,dash pattern=on 4pt off 2pt ,line join=round] ( 44.70,114.62) --
	( 45.58,114.99) --
	( 46.46,115.06) --
	( 47.34,115.16) --
	( 48.22,115.24) --
	( 49.11,115.22) --
	( 49.99,115.41) --
	( 50.87,114.94) --
	( 51.75,114.26) --
	( 52.64,114.38) --
	( 53.52,114.65) --
	( 54.40,115.09) --
	( 55.28,115.70) --
	( 56.16,115.44) --
	( 57.05,115.64) --
	( 57.93,115.52) --
	( 58.81,115.65) --
	( 59.69,115.88) --
	( 60.58,115.83) --
	( 61.46,115.46) --
	( 62.34,115.56) --
	( 63.22,115.64) --
	( 64.11,114.43) --
	( 64.99,115.71) --
	( 65.87,115.35) --
	( 66.75,115.82) --
	( 67.63,115.89) --
	( 68.52,115.93) --
	( 69.40,115.91) --
	( 70.28,115.69) --
	( 71.16,115.72) --
	( 72.05,116.18) --
	( 72.93,116.07) --
	( 73.81,115.51) --
	( 74.69,115.82) --
	( 75.57,116.17) --
	( 76.46,115.50) --
	( 77.34,115.05) --
	( 78.22,115.31) --
	( 79.10,115.62) --
	( 79.99,116.16) --
	( 80.87,115.33) --
	( 81.75,115.84) --
	( 82.63,116.41) --
	( 83.52,115.94) --
	( 84.40,116.10) --
	( 85.28,115.75) --
	( 86.16,115.59) --
	( 87.04,115.61) --
	( 87.93,116.37) --
	( 88.81,116.07) --
	( 89.69,116.28) --
	( 90.57,116.04) --
	( 91.46,115.99) --
	( 92.34,115.86) --
	( 93.22,115.59) --
	( 94.10,116.27) --
	( 94.98,116.20) --
	( 95.87,116.55) --
	( 96.75,116.30) --
	( 97.63,116.13) --
	( 98.51,115.77) --
	( 99.40,116.41) --
	(100.28,115.90) --
	(101.16,116.43) --
	(102.04,116.42) --
	(102.92,116.20) --
	(103.81,116.17) --
	(104.69,116.42) --
	(105.57,115.83) --
	(106.45,116.32) --
	(107.34,116.44) --
	(108.22,116.56) --
	(109.10,116.18) --
	(109.98,116.04) --
	(110.87,115.97) --
	(111.75,115.69) --
	(112.63,176.71) --
	(113.51,176.94) --
	(114.39,178.47) --
	(115.28,178.90) --
	(116.16,179.49) --
	(117.04,180.28) --
	(117.92,181.25) --
	(118.81,181.52) --
	(119.69,182.45) --
	(120.57,183.19) --
	(121.45,183.96) --
	(122.33,185.34) --
	(123.22,186.41) --
	(124.10,187.37) --
	(124.98,188.24) --
	(125.86,189.04) --
	(126.75,189.66) --
	(127.63,190.55) --
	(128.51,191.18) --
	(129.39,191.66) --
	(130.28,192.09) --
	(131.16,193.41) --
	(132.04,194.38) --
	(132.92,195.28) --
	(133.80,196.06) --
	(134.69,196.79) --
	(135.57,197.52) --
	(136.45,198.35) --
	(137.33,199.15) --
	(138.22,200.02) --
	(139.10,200.83) --
	(139.98,201.62) --
	(140.86,202.33) --
	(141.74,202.86) --
	(142.63,201.10) --
	(143.51,204.34) --
	(144.39,205.21) --
	(145.27,205.92) --
	(146.16,206.00) --
	(147.04,207.16) --
	(147.92,208.07) --
	(148.80,208.55) --
	(149.69,209.57) --
	(150.57,210.54) --
	(151.45,211.32) --
	(152.33,211.61) --
	(153.21,212.71) --
	(154.10,213.81) --
	(154.98,214.73) --
	(155.86,215.44) --
	(156.74,216.31) --
	(157.63,216.86) --
	(158.51,217.09) --
	(159.39,217.46) --
	(160.27,219.09) --
	(161.15,220.00) --
	(162.04,220.20) --
	(162.92,221.25) --
	(163.80,222.46) --
	(164.68,223.33) --
	(165.57,223.99) --
	(166.45,224.76) --
	(167.33,225.64) --
	(168.21,226.58) --
	(169.10,227.47) --
	(169.98,228.31) --
	(170.86,229.08) --
	(171.74,229.79) --
	(172.62,230.46) --
	(173.51,231.13) --
	(174.39,231.83) --
	(175.27,232.29) --
	(176.15,232.88) --
	(177.04,234.10) --
	(177.92,234.22) --
	(178.80,235.53) --
	(179.68,236.52) --
	(180.56,237.53) --
	(181.45,238.51) --
	(182.33,239.41) --
	(183.21,240.28) --
	(184.09,241.13) --
	(184.98,241.98) --
	(185.86,242.73) --
	(186.74,243.59) --
	(187.62,244.34) --
	(188.50,245.04) --
	(189.39,245.67) --
	(190.27,246.19) --
	(191.15,246.22) --
	(192.03,247.41) --
	(192.92,248.36) --
	(193.80,248.99) --
	(194.68,249.21) --
	(195.56,250.80) --
	(196.45,251.73) --
	(197.33,252.61) --
	(198.21,253.41) --
	(199.09,254.10) --
	(199.97,254.80) --
	(200.86,255.45) --
	(201.74,256.13) --
	(202.62,256.90) --
	(203.50,257.65) --
	(204.39,258.40) --
	(205.27,259.10) --
	(206.15,259.95) --
	(207.03,260.85) --
	(207.91,261.75) --
	(208.80,262.68) --
	(209.68,263.52) --
	(210.56,264.07) --
	(211.44,265.09) --
	(212.33,265.61) --
	(213.21,266.68) --
	(214.09,267.41) --
	(214.97,268.35) --
	(215.86,269.22) --
	(216.74,270.10) --
	(217.62,270.94) --
	(218.50,271.75) --
	(219.38,272.54) --
	(220.27,273.32);
\end{scope}
\begin{scope}
\definecolor{drawColor}{RGB}{190,190,190}

\path[draw=drawColor,line width= 0.6pt,line join=round] ( 35.92,105.26) --
	( 35.92,281.32);
\end{scope}
\begin{scope}
\definecolor{drawColor}{gray}{0.30}

\node[text=drawColor,anchor=base east,inner sep=0pt, outer sep=0pt, scale=  0.88] at ( 30.97,125.23) {0};

\node[text=drawColor,anchor=base east,inner sep=0pt, outer sep=0pt, scale=  0.88] at ( 30.97,169.34) {50};

\node[text=drawColor,anchor=base east,inner sep=0pt, outer sep=0pt, scale=  0.88] at ( 30.97,213.45) {100};

\node[text=drawColor,anchor=base east,inner sep=0pt, outer sep=0pt, scale=  0.88] at ( 30.97,257.57) {150};
\end{scope}
\begin{scope}
\definecolor{drawColor}{gray}{0.20}

\path[draw=drawColor,line width= 0.6pt,line join=round] ( 33.17,128.26) --
	( 35.92,128.26);

\path[draw=drawColor,line width= 0.6pt,line join=round] ( 33.17,172.37) --
	( 35.92,172.37);

\path[draw=drawColor,line width= 0.6pt,line join=round] ( 33.17,216.48) --
	( 35.92,216.48);

\path[draw=drawColor,line width= 0.6pt,line join=round] ( 33.17,260.60) --
	( 35.92,260.60);
\end{scope}
\begin{scope}
\definecolor{drawColor}{RGB}{190,190,190}

\path[draw=drawColor,line width= 0.6pt,line join=round] ( 35.92,105.26) --
	(229.05,105.26);
\end{scope}
\begin{scope}
\definecolor{drawColor}{gray}{0.20}

\path[draw=drawColor,line width= 0.6pt,line join=round] ( 42.93,102.51) --
	( 42.93,105.26);

\path[draw=drawColor,line width= 0.6pt,line join=round] ( 87.04,102.51) --
	( 87.04,105.26);

\path[draw=drawColor,line width= 0.6pt,line join=round] (131.16,102.51) --
	(131.16,105.26);

\path[draw=drawColor,line width= 0.6pt,line join=round] (175.27,102.51) --
	(175.27,105.26);

\path[draw=drawColor,line width= 0.6pt,line join=round] (219.38,102.51) --
	(219.38,105.26);
\end{scope}
\begin{scope}
\definecolor{drawColor}{gray}{0.30}

\node[text=drawColor,anchor=base,inner sep=0pt, outer sep=0pt, scale=  0.88] at ( 42.93, 94.24) {0};

\node[text=drawColor,anchor=base,inner sep=0pt, outer sep=0pt, scale=  0.88] at ( 87.04, 94.24) {50};

\node[text=drawColor,anchor=base,inner sep=0pt, outer sep=0pt, scale=  0.88] at (131.16, 94.24) {100};

\node[text=drawColor,anchor=base,inner sep=0pt, outer sep=0pt, scale=  0.88] at (175.27, 94.24) {150};

\node[text=drawColor,anchor=base,inner sep=0pt, outer sep=0pt, scale=  0.88] at (219.38, 94.24) {200};
\end{scope}
\begin{scope}
\definecolor{drawColor}{RGB}{0,0,0}

\node[text=drawColor,anchor=base,inner sep=0pt, outer sep=0pt, scale=  1.10] at (132.48, 81.97) {$n$};
\end{scope}
\begin{scope}
\definecolor{drawColor}{RGB}{0,0,0}

\node[text=drawColor,rotate= 90.00,anchor=base,inner sep=0pt, outer sep=0pt, scale=  1.10] at ( 13.08,193.29) {Magnitude of the relative error};
\end{scope}
\begin{scope}
\definecolor{fillColor}{RGB}{255,255,255}

\path[fill=fillColor] (240.05,158.60) rectangle (355.85,227.98);
\end{scope}
\begin{scope}
\definecolor{drawColor}{RGB}{0,0,0}

\node[text=drawColor,anchor=base west,inner sep=0pt, outer sep=0pt, scale=  1.10] at (245.55,213.93) {Algorithm};
\end{scope}
\begin{scope}
\definecolor{drawColor}{RGB}{255,255,255}
\definecolor{fillColor}{gray}{0.95}

\path[draw=drawColor,line width= 0.6pt,line join=round,line cap=round,fill=fillColor] (245.55,193.00) rectangle (260.00,207.46);
\end{scope}
\begin{scope}
\definecolor{drawColor}{gray}{0.20}

\path[draw=drawColor,line width= 1.1pt,line join=round] (246.99,200.23) -- (258.55,200.23);
\end{scope}
\begin{scope}
\definecolor{drawColor}{RGB}{255,255,255}
\definecolor{fillColor}{gray}{0.95}

\path[draw=drawColor,line width= 0.6pt,line join=round,line cap=round,fill=fillColor] (245.55,178.55) rectangle (260.00,193.00);
\end{scope}
\begin{scope}
\definecolor{drawColor}{RGB}{152,152,152}

\path[draw=drawColor,line width= 1.1pt,dash pattern=on 2pt off 2pt ,line join=round] (246.99,185.78) -- (258.55,185.78);
\end{scope}
\begin{scope}
\definecolor{drawColor}{RGB}{255,255,255}
\definecolor{fillColor}{gray}{0.95}

\path[draw=drawColor,line width= 0.6pt,line join=round,line cap=round,fill=fillColor] (245.55,164.10) rectangle (260.00,178.55);
\end{scope}
\begin{scope}
\definecolor{drawColor}{gray}{0.80}

\path[draw=drawColor,line width= 1.1pt,dash pattern=on 4pt off 2pt ,line join=round] (246.99,171.32) -- (258.55,171.32);
\end{scope}
\begin{scope}
\definecolor{drawColor}{RGB}{0,0,0}

\node[text=drawColor,anchor=base west,inner sep=0pt, outer sep=0pt, scale=  0.88] at (265.50,197.20) {Poisson O($n^2$)};
\end{scope}
\begin{scope}
\definecolor{drawColor}{RGB}{0,0,0}

\node[text=drawColor,anchor=base west,inner sep=0pt, outer sep=0pt, scale=  0.88] at (265.50,182.75) {Numerical Integration};
\end{scope}
\begin{scope}
\definecolor{drawColor}{RGB}{0,0,0}

\node[text=drawColor,anchor=base west,inner sep=0pt, outer sep=0pt, scale=  0.88] at (265.50,168.29) {Poisson O($n^2\log(n)$)};
\end{scope}
\end{tikzpicture}

%% file: 3-generalized-recursions.tex
\section{The Generalized Recursions} \label{sec3}

Let $n:=n_1+n_2$, $n_1,n_2\in\mathbb{N}$. Furthermore, let $X_{1}, \hdots, X_{n_1}\sim \text{Uni}[0,1]$ and 
$X_{n_1+1}, \hdots, X_{n_1+n_2}\sim F$ be jointly stochastically independent. Let for $0\leq i_1\leq n_1$ and $0\leq i_2\leq n_2$\begin{align}
    \Psi(i_1,i_2)&:=\mathbb{P}\left(X_{1:M}\leq b_1,\hdots,X_{i:M}\leq b_{i_1+i_2}\right)\label{eq:two_group_order_statistic},
\end{align} 
where $M:=[i_1]\bigcup\left\{n_1+j\left|j\in[i_2]\right.\right\}$, $[0]:=\emptyset$, $X_{i:M}$ denotes the 
$i$-th order statistic of $(X_j)_{j\in M}$, and $(b_i)_{i\in[n]}$ is an increasing sequence with values in 
$[0,1]$. The following subsections provide formulas for efficiently calculating $\Psi(n_1, n_2)$, and we discuss their computational and numerical properties.

\subsection{Generalization of Bolshev's Recursion} \label{sec31}
\begin{lemma}[Generalization of Bolshev's Recursion]\label{lemma:generalized_bolshev_recursion}
The function $\Psi$ from \eqref{eq:two_group_order_statistic} satisfies the recursion
\begin{align*}
    \Psi(m_1,m_2)&=1-\sum\limits_{\substack{0\leq k_1\leq m_1\\0\leq k_2\leq m_2\\k_1+k_2<m_1+m_2}}M^{(m_1,m_2)}_{k_1,k_2}\cdot\Psi(k_1,k_2),
\end{align*}
where
\begin{equation}
    M^{(m_1,m_2)}_{k_1,k_2} :=\binom{m_1}{k_1}\binom{m_2}{k_2}(1-b_{k_1+k_2+1})^{m_1-k_1}\cdot(1-F(b_{k_1+k_2+1}))^{m_2-k_2}.\label{eq:m_general_bolshev}
\end{equation}
Moreover, we have the following recursive relationships for $M$.
\begin{align*}
    M^{(m_1+1,m_2)}_{k_1,k_2}&=\begin{cases}
        1& k_2=m_2\land k_1=m_1+1\\
        M^{(m_1,m_2)}_{m_1,k_2+1}\cdot\frac{k_2+1}{m_2-k_2}\cdot(1-F(b_{m_1+(k_2+1)+1}))&k_2<m_2\land k_1=m_1+1\\
        M^{(m_1,m_2)}_{k_1,k_2}\cdot\frac{(m_1+1)}{m_1+1-k_1}\cdot(1-b_{k_1+k_2+1})&\text{otherwise},
    \end{cases}\\
    M^{(m_1,m_2+1)}_{k_1,k_2}&=\begin{cases}
        1&k_1=m_1\land k_2=m_2+1\\
        M^{(m_1,m_2)}_{k_1+1,m_2}\cdot\frac{k_1+1}{m_1-k_1}\cdot(1-b_{k_1+m_2+2})&k_1<m_1\land k_2=m_2+1\\
        M^{(m_1,m_2)}_{k_1,k_2}\cdot\frac{(m_2+1)}{m_2+1-k_2}\cdot(1-F(b_{k_1+k_2+1}))&\text{otherwise}.
    \end{cases}
\end{align*}
\end{lemma}

For $n_1=0$ or $n_2=0$ this is simply the well-known Bolshev recursion.

\subsection{Generalization of Steck's Recursion}
\begin{lemma}[Generalization of Steck's Recursion]
Let $b_0:=0$. Then $\Psi$ from \eqref{eq:two_group_order_statistic} satisfies the recursion
\begin{align*}
    \Psi(m_1,m_2)&=(b_{m_1+m_2})^{m_1}F(b_{m_1+m_2})^{m_2}-\sum_{\substack{0\leq k_1\leq m_1\\
    0\leq k_2\leq m_2\\
    k_1+k_2\leq m_1+m_2-2}}M^{(m_1,m_2)}_{k_1,k_2}\cdot\Psi(k_1,k_2),
\end{align*}
where		
\begin{equation}
    M^{(m_1,m_2)}_{k_1,k_2} :=\binom{m_1}{k_1}\binom{m_2}{k_2}\left(b_{m_1+m_2}-b_{k_1+k_2+1}\right)^{m_1-k_1}\left(F(b_{m_1+m_2})-F(b_{k_1+k_2+1})\right)^{m_2-k_2}.\label{eq:steck_rec_eq}
\end{equation}
Letting
\[
    a(k,j) :=\binom{k}{j} \text{~~and~~}
    a(k,j) =\begin{cases}
        1&k=j\\
        \frac{j+1}{k-j}\times a(k,j+1)&j<k,
    \end{cases}\numberthis\label{eq:steck_binomial_rec}
\]
we can write
\begin{align*}		
    M^{(m_1,m_2)}_{0,j}&=a(m_2,j)\cdot\left(b_{m_1+m_2}-b_{j+1}\right)^{m_1}\left(F(b_{m_1+m_2})-F(b_{j+1})\right)^{m_2-j},\numberthis\label{eq:steck_coeff1}\\
    M^{(m_1,m_2)}_{j,m_2}&=a(m_1,j)\cdot\left(b_{m_1+m_2}-b_{j+m_2+1}\right)^{m_1-j}.\numberthis\label{eq:steck_coeff2}
\end{align*}		
Furthermore, we have the following recursion for $M$.		
\begin{align*}
    M^{(m_1,m_2)}_{k_1+1,k_2-1}&=M^{(m_1,m_2)}_{k_1,k_2}\times\frac{F(b_{m_1+m_2})-F(b_{k_1+k_2+1})}{b_{m_1+m_2}-b_{k_1+k_2+1}}\times \frac{m_1-k_1}{k_1+1}\times
    \frac{m_2-k_2+1}{k_2}
\end{align*}

for $0\leq m_1\leq n_1$ and $0\leq m_2\leq n_2$.
\end{lemma}
\begin{proof}
    See \cite[Proposition 1]{blanchard2014least}
\end{proof}

\subsection{Generalization of Noe's Recursion} \label{sec33}
\begin{lemma}[Generalization of Noe's Recursion]\label{lemma:generalized_noe_rec}
Let $b_0:=0$, $Q_{0,0}(0):=1$, $Q_{i_1,i_2}(1):=b_1^{i_1}\cdot F(b_1)^{i_2}$ and for $m>1$
\begin{align*}
    Q_{i_1,i_2}(m)&:=\sum_{\substack{0\leq k_1\leq i_1\\0\leq k_2\leq i_2\\ m-1\leq k_1+k_2}}M^{i_1,i_2}_{k_1,k_2}(m)\cdot Q_{k_1,k_2}(m-1),\\
    M^{i_1,i_2}_{k_1,k_2}(m)&:=\binom{i_1}{k_1}\binom{i_2}{k_2}\times(b_m-b_{m-1})^{i_1-k_1}\times (F(b_{m})-F(b_{m-1}))^{i_2-k_2}
\end{align*}
for $0\leq i_1\leq n_1,0\leq i_2\leq n_2,m\leq i_1+i_2\leq n$.

Then the function $\Psi$ from \eqref{eq:two_group_order_statistic} satisfies
\begin{align*}
    \Psi(i_1,i_2)&=Q_{i_1,i_2}(i_1+i_2)
\end{align*}
for $i_1 \leq n_1$ and $i_2 \leq n_2$.

Letting
\[	
a^{(m),1}(j) :=(b_m-b_{m-1})^{j} \text{~~and~~}
    a^{(m),2}(j) :=\left(F(b_m)-F(b_{m-1})\right)^{j},
\]
we can write
\begin{equation}
    M^{i_1,i_2}_{k_1,k_2}(m) = \binom{i_1}{k_1}\binom{i_2}{k_2}\times a^{(m),1}(i_1-k_1)\times a^{(m),2}(i_2-k_2). \label{eq:noe_coeff_rec}
\end{equation}   
\end{lemma}

\subsection{Computational Complexity and Numerical Properties}

The computational complexity (defined to be the number of elementary arithmetic operations on floating point numbers) of each of the aforementioned recursions is given by the following lemma.
\begin{lemma}
The proposed recursions can be implemented using
\begin{description}[leftmargin=!,labelwidth=\widthof{\bfseries Bolshev}]\label{lemma:rec_complexity}
    \item[Bolshev] $O(n_1^2n_2^2)$
    \item[Steck] $O\left(n_1^2n_2^2\log_2(n_1n_2)\right)$
    \item[Noe] $O\left(n_1^2n_2^2(n_1+n_2)\right)$
\end{description} elementary arithmetic operations (addition, subtraction, multiplication, division) and $O(n_1n_2)$ memory (assuming fixed-precision storage of all results).
\end{lemma}

The results of Lemma \ref{lemma:rec_complexity} suggest that Noe's recursion might not be the best choice. However, for small values of the $b_i$'s, Bolshev's recursion and Steck's recursion are inherently numerically unstable. Consider for example $n_1=11,n_2=0$ and 
\begin{align*}
    b_i&:=\begin{cases}2^{-10}&\text{if }i\leq 10,\\
    2^{-1}&\text{if }i=11.
		\end{cases}
\end{align*} 
Then both recursions, when implemented in double precision floating point arithmetic, result in negative values and huge relative errors (cf. Table \ref{tbl:steck_bolshev}) which can be explained by inaccurate or catastrophic cancellation, respectively.

\input{./bolshev_steck}

Noe's recursion, if implemented in a reasonable manner, never results in negative values. Furthermore by \cite[Equation (3)]{jeannerod2018relative} the relative error is bounded (if the coefficients are computed with a bounded relative error) since all summands are non-negative.

\begin{remark}
    Noe's recursion can be easily parallelized since the $Q_{i_1,i_2}(m)$ can be, for any fixed $m$, computed in parallel.
\end{remark}

%% file: bolshev_steck.tex
\begin{table}[ht]
\centering
\begin{tabular}{lrrrrr}
  \hline
 & Exact Probability & Steck & Rel. Err. (Steck) & Bolshev & Rel. Err. (Bolshev) \\ 
  \hline
2 & 9.76562E-04 & 9.76562E-04 & 0.00000E+00 & 9.76562E-04 & 0.00000E+00 \\ 
  3 & 9.53674E-07 & 9.53674E-07 & 0.00000E+00 & 9.53674E-07 & 0.00000E+00 \\ 
  4 & 9.31323E-10 & 9.31323E-10 & 0.00000E+00 & 9.31323E-10 & 0.00000E+00 \\ 
  5 & 9.09495E-13 & 9.09495E-13 & 0.00000E+00 & 9.09495E-13 & 0.00000E+00 \\ 
  6 & 8.88178E-16 & 8.88178E-16 & 0.00000E+00 & 8.88178E-16 & 0.00000E+00 \\ 
  7 & 8.67362E-19 & 8.67362E-19 & 0.00000E+00 & 1.73472E-18 & 1.00000E+00 \\ 
  8 & 8.47033E-22 & 8.47033E-22 & 0.00000E+00 & 1.10114E-20 & 1.20000E+01 \\ 
  9 & 8.27181E-25 & 8.27181E-25 & 0.00000E+00 & 6.85071E-21 & 8.28100E+03 \\ 
  10 & 8.07794E-28 & 8.07794E-28 & 0.00000E+00 & -2.70517E-20 & 3.34884E+07 \\ 
  11 & 7.88861E-31 & 7.88861E-31 & 0.00000E+00 & -1.12683E-16 & 1.42842E+14 \\ 
  12 & 4.33103E-30 & -1.75898E-20 & 4.06134E+09 & 2.83880E-16 & 6.55456E+13 \\ 
   \hline
\end{tabular}
\caption{Calculation of the probability that uniform order statistics are bounded above by $b_i\in[0,1]^11$ where $b_i=2^{-10}$ if $i\leq10$ and $b_{11}=2^{-1}$. The rows give the intermediate steps of the algorithms. The first column reports the first few non-zero digits of the exact probabilities (computed in rational arithmetic), the second column reports the steps of Steck's recursion (calculated in double precision floating point arithmetic) and the fourth column reports the steps of Bolshev's recursion (also calculated in double precision floating point arithmetic). In the third and the fifth column the relative error of the intermediate value is reported.} 
\label{tbl:steck_bolshev}
\end{table}

%% file: 4-exact-bolshev.tex
\section{Exact Evaluation of Bolshev's Recursion}\label{sec:exact_bolshev}
If only elementary arithmetic operations are utilized and the number of such operations is not too large it is feasible to exactly evaluate expressions using rational arithmetic.\footnote{Our C++ implementation is based on \href{https://gmplib.org/}{The GNU Multiple Precision Arithmetic Library}.} We show that this is indeed the case for the Bolshev recursion as well its generalization for the two-group case presented in Section \ref{sec31}.

First we consider the case where $n_2=0$, hence $n_1 = n$: Even though Bolshev's recursion involves binomial coefficients 
our proposed Algorithm \ref{alg:bolshev} (cf.\ the Appendix) for the one-group case evaluates it using only
\[
    \text{\#Operations} = n+(n+1)+\sum_{k=2}^n\left[2+\sum_{j=1}^{k-1}6\right]=3n^2+n-1
\] 
elementary arithmetic operations (addition, subtraction, multiplication, division). For an illustration, considering the sequence $b_i \equiv b_i^{(n)}:=0.05\times i / n$ we observed the execution times depicted in Figure \ref{fig:bolshev_single}.

\begin{figure}[!htb]
	\centering
    \resizebox{.5\linewidth}{!}{\input{./bolshev_single_execution_time}}
	\caption{Execution time of Algorithm \ref{alg:bolshev} in rational arithmetic.}
	\label{fig:bolshev_single}
\end{figure}
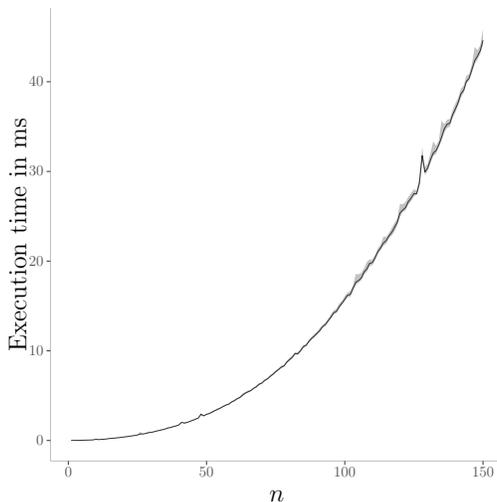

\begin{remark}
Our proposed Algorithm \ref{alg:bolshev_general} (cf.\ the Appendix) implements the two-group case 
in $O(n_1^2n_2^2)$ elementary arithmetic operations.
Consequently, for equal sample sizes $n_1=n_2=\ell$ the number of operations is of $\operatorname{O}(\ell^4)$. Notice that this is a marked improvement over the exponential complexity $\ell^\ell$ reported by \cite[Theorem 4.2]{glueck2008fast}. Figure \ref{fig:comp_bolshev_noe} illustrates the observed execution time for calculating $f(\ell):=\Psi(\ell,\ell)$ and $b_i \equiv b_i^{(n)}:=0.05 \times i / n$, where $n = n_1 + n_2 = 2 \ell$. 

For not necessarily equal sample sizes $n_1 \geq 1$ and $n_2 \geq 1$, our implementation of Algorithm \ref{alg:bolshev_general} needs 
$f(n_1,n_2):= 1.5 \cdot n_1^2 \cdot n_2^2 + 4.5 \cdot (n_1^2 \cdot n_2 + n_1 \cdot n_2^2) + 
3 \cdot (n_1 + n_1^2 + n_2 + n_2^2) + 7.5 \cdot n_1 \cdot n_2 + 2$ arithmetic operations.
\end{remark}

Since the cdf $F$ of many interesting distributions is not available in a closed form the thresholds $F(b_i)$ might either not be exactly calculable or simply not exactly representable as rational numbers. 
Lemma \ref{lemma:inexact_thresholds} analyzes the error propagation when $F$ and / or $\mathbf{b}$ are inexact.

\begin{lemma}\label{lemma:inexact_thresholds}
	Let \begin{equation}
		x_i:=\begin{cases}
			b_1&\text{if }i=1\\
			b_{i}-b_{i-1}&\text{if }1<i\leq n\\
			F(b_1)&\text{if }i=n+1\\
			F(b_{i})-F(b_{i-1})&\text{if }n+1<i\leq 2n
		\end{cases}\label{eq:inexact_thresholds_x_def}
	\end{equation} and denote by $\left(\tilde x_i\right)_{i\in[n]}$ approximations thereof, which are obtained by replacing $\left(b_i\right)_{i\in[n]}$ and $F$ by approximations $(\tilde b_i)_{i\in[n]}$ and $\tilde F$.
	If for $\varepsilon\in(0,1)$ it holds that for $\forall i\in[2n]:x_i\in(1-\varepsilon,1+\varepsilon)$ the it follows that for all $(i_1,i_2)\in[n_1]\times[n_2]$ 
	\begin{equation*}
        \tilde\Psi(i_1,i_2)\in \Psi(i_1,i_2)\cdot\left((1-\varepsilon)^{i_1+i_2},(1+\varepsilon)^{i_1+i_2}\right),
    \end{equation*} 
		where $\tilde\Psi$ denotes the approximation of $\Psi$ obtained by using $\tilde F$ and $\tilde b_i$ instead of $F$ and $b_i$.
\end{lemma}

%% file: bolshev_single_execution_time.tex
\begin{tikzpicture}[x=1pt,y=1pt]
\definecolor{fillColor}{RGB}{255,255,255}
\begin{scope}
\definecolor{drawColor}{RGB}{255,255,255}
\definecolor{fillColor}{RGB}{255,255,255}

\path[draw=drawColor,line width= 0.6pt,line join=round,line cap=round,fill=fillColor] (  0.00,  0.40) rectangle (361.35,360.95);
\end{scope}
\begin{scope}
\definecolor{drawColor}{RGB}{255,255,255}

\path[draw=drawColor,line width= 0.3pt,line join=round] ( 34.27, 80.30) --
	(355.85, 80.30);

\path[draw=drawColor,line width= 0.3pt,line join=round] ( 34.27,144.00) --
	(355.85,144.00);

\path[draw=drawColor,line width= 0.3pt,line join=round] ( 34.27,207.69) --
	(355.85,207.69);

\path[draw=drawColor,line width= 0.3pt,line join=round] ( 34.27,271.39) --
	(355.85,271.39);

\path[draw=drawColor,line width= 0.3pt,line join=round] ( 34.27,335.09) --
	(355.85,335.09);

\path[draw=drawColor,line width= 0.3pt,line join=round] ( 95.98, 33.88) --
	( 95.98,355.45);

\path[draw=drawColor,line width= 0.3pt,line join=round] (194.08, 33.88) --
	(194.08,355.45);

\path[draw=drawColor,line width= 0.3pt,line join=round] (292.18, 33.88) --
	(292.18,355.45);

\path[draw=drawColor,line width= 0.6pt,line join=round] ( 34.27, 48.45) --
	(355.85, 48.45);

\path[draw=drawColor,line width= 0.6pt,line join=round] ( 34.27,112.15) --
	(355.85,112.15);

\path[draw=drawColor,line width= 0.6pt,line join=round] ( 34.27,175.85) --
	(355.85,175.85);

\path[draw=drawColor,line width= 0.6pt,line join=round] ( 34.27,239.54) --
	(355.85,239.54);

\path[draw=drawColor,line width= 0.6pt,line join=round] ( 34.27,303.24) --
	(355.85,303.24);

\path[draw=drawColor,line width= 0.6pt,line join=round] ( 46.93, 33.88) --
	( 46.93,355.45);

\path[draw=drawColor,line width= 0.6pt,line join=round] (145.03, 33.88) --
	(145.03,355.45);

\path[draw=drawColor,line width= 0.6pt,line join=round] (243.13, 33.88) --
	(243.13,355.45);

\path[draw=drawColor,line width= 0.6pt,line join=round] (341.23, 33.88) --
	(341.23,355.45);
\definecolor{fillColor}{RGB}{51,51,51}

\path[fill=fillColor,fill opacity=0.30] ( 48.89, 48.50) --
	( 50.85, 48.52) --
	( 52.81, 48.55) --
	( 54.78, 48.59) --
	( 56.74, 48.63) --
	( 58.70, 48.69) --
	( 60.66, 48.76) --
	( 62.62, 48.84) --
	( 64.59, 48.93) --
	( 66.55, 49.38) --
	( 68.51, 49.16) --
	( 70.47, 49.30) --
	( 72.43, 49.46) --
	( 74.40, 49.66) --
	( 76.36, 50.41) --
	( 78.32, 50.08) --
	( 80.28, 50.28) --
	( 82.24, 50.49) --
	( 84.21, 50.75) --
	( 86.17, 50.97) --
	( 88.13, 51.24) --
	( 90.09, 51.53) --
	( 92.05, 51.83) --
	( 94.02, 52.15) --
	( 95.98, 52.47) --
	( 97.94, 54.67) --
	( 99.90, 53.22) --
	(101.86, 53.61) --
	(103.83, 54.90) --
	(105.79, 54.41) --
	(107.75, 54.86) --
	(109.71, 55.47) --
	(111.67, 55.80) --
	(113.64, 56.38) --
	(115.60, 56.83) --
	(117.56, 57.73) --
	(119.52, 57.90) --
	(121.49, 58.67) --
	(123.45, 59.01) --
	(125.41, 59.85) --
	(127.37, 61.66) --
	(129.33, 61.30) --
	(131.30, 61.49) --
	(133.26, 62.31) --
	(135.22, 63.10) --
	(137.18, 63.87) --
	(139.14, 64.67) --
	(141.11, 67.88) --
	(143.07, 66.27) --
	(145.03, 67.38) --
	(146.99, 67.95) --
	(148.95, 69.44) --
	(150.92, 70.12) --
	(152.88, 71.35) --
	(154.84, 71.97) --
	(156.80, 73.15) --
	(158.76, 74.14) --
	(160.73, 74.80) --
	(162.69, 76.45) --
	(164.65, 77.37) --
	(166.61, 78.71) --
	(168.57, 79.57) --
	(170.54, 81.90) --
	(172.50, 82.87) --
	(174.46, 83.96) --
	(176.42, 84.55) --
	(178.38, 86.14) --
	(180.35, 87.52) --
	(182.31, 89.24) --
	(184.27, 89.84) --
	(186.23, 92.03) --
	(188.19, 92.83) --
	(190.16, 95.04) --
	(192.12, 96.28) --
	(194.08, 98.27) --
	(196.04, 99.52) --
	(198.00,101.58) --
	(199.97,102.22) --
	(201.93,105.16) --
	(203.89,107.12) --
	(205.85,108.65) --
	(207.81,111.24) --
	(209.78,111.15) --
	(211.74,113.36) --
	(213.70,116.38) --
	(215.66,116.89) --
	(217.62,120.01) --
	(219.59,122.04) --
	(221.55,124.21) --
	(223.51,125.87) --
	(225.47,128.60) --
	(227.43,130.42) --
	(229.40,132.33) --
	(231.36,134.94) --
	(233.32,138.12) --
	(235.28,140.87) --
	(237.25,142.16) --
	(239.21,145.63) --
	(241.17,147.74) --
	(243.13,151.15) --
	(245.09,153.44) --
	(247.06,154.69) --
	(249.02,158.88) --
	(250.98,166.74) --
	(252.94,166.47) --
	(254.90,167.54) --
	(256.87,170.79) --
	(258.83,175.38) --
	(260.79,177.07) --
	(262.75,177.12) --
	(264.71,180.65) --
	(266.68,185.01) --
	(268.64,187.82) --
	(270.60,193.48) --
	(272.56,192.98) --
	(274.52,195.80) --
	(276.49,199.50) --
	(278.45,202.84) --
	(280.41,206.36) --
	(282.37,216.55) --
	(284.33,216.25) --
	(286.30,217.94) --
	(288.26,221.45) --
	(290.22,224.17) --
	(292.18,226.87) --
	(294.14,226.01) --
	(296.11,235.10) --
	(298.07,256.66) --
	(300.03,242.34) --
	(301.99,245.29) --
	(303.95,252.48) --
	(305.92,261.18) --
	(307.88,257.72) --
	(309.84,261.25) --
	(311.80,275.57) --
	(313.76,273.16) --
	(315.73,276.26) --
	(317.69,276.88) --
	(319.65,282.95) --
	(321.61,287.76) --
	(323.57,291.72) --
	(325.54,298.33) --
	(327.50,300.56) --
	(329.46,307.70) --
	(331.42,308.71) --
	(333.38,315.83) --
	(335.35,327.99) --
	(337.31,325.65) --
	(339.27,329.07) --
	(341.23,340.84) --
	(341.23,329.39) --
	(339.27,322.92) --
	(337.31,319.08) --
	(335.35,315.75) --
	(333.38,309.48) --
	(331.42,304.28) --
	(329.46,301.88) --
	(327.50,295.34) --
	(325.54,291.88) --
	(323.57,286.33) --
	(321.61,281.41) --
	(319.65,277.44) --
	(317.69,271.34) --
	(315.73,270.44) --
	(313.76,267.14) --
	(311.80,261.14) --
	(309.84,256.81) --
	(307.88,251.87) --
	(305.92,250.39) --
	(303.95,245.28) --
	(301.99,239.92) --
	(300.03,236.91) --
	(298.07,248.37) --
	(296.11,229.34) --
	(294.14,222.58) --
	(292.18,221.98) --
	(290.22,218.43) --
	(288.26,216.11) --
	(286.30,211.81) --
	(284.33,210.62) --
	(282.37,207.37) --
	(280.41,201.18) --
	(278.45,196.82) --
	(276.49,193.95) --
	(274.52,192.12) --
	(272.56,188.19) --
	(270.60,186.85) --
	(268.64,183.62) --
	(266.68,181.33) --
	(264.71,176.66) --
	(262.75,172.98) --
	(260.79,172.29) --
	(258.83,168.36) --
	(256.87,166.19) --
	(254.90,162.37) --
	(252.94,160.54) --
	(250.98,158.86) --
	(249.02,155.12) --
	(247.06,150.34) --
	(245.09,150.29) --
	(243.13,147.53) --
	(241.17,144.70) --
	(239.21,142.05) --
	(237.25,139.06) --
	(235.28,137.90) --
	(233.32,134.90) --
	(231.36,132.54) --
	(229.40,129.84) --
	(227.43,128.34) --
	(225.47,125.74) --
	(223.51,123.78) --
	(221.55,121.61) --
	(219.59,120.03) --
	(217.62,117.81) --
	(215.66,115.35) --
	(213.70,114.14) --
	(211.74,111.36) --
	(209.78,109.44) --
	(207.81,109.33) --
	(205.85,106.61) --
	(203.89,105.11) --
	(201.93,103.28) --
	(199.97,100.78) --
	(198.00, 99.79) --
	(196.04, 98.12) --
	(194.08, 96.53) --
	(192.12, 94.88) --
	(190.16, 93.47) --
	(188.19, 91.68) --
	(186.23, 90.47) --
	(184.27, 88.78) --
	(182.31, 87.63) --
	(180.35, 86.11) --
	(178.38, 84.99) --
	(176.42, 83.43) --
	(174.46, 82.44) --
	(172.50, 81.55) --
	(170.54, 79.74) --
	(168.57, 78.41) --
	(166.61, 77.59) --
	(164.65, 76.23) --
	(162.69, 75.23) --
	(160.73, 73.89) --
	(158.76, 73.26) --
	(156.80, 72.21) --
	(154.84, 71.02) --
	(152.88, 70.21) --
	(150.92, 69.25) --
	(148.95, 68.11) --
	(146.99, 67.25) --
	(145.03, 66.76) --
	(143.07, 65.70) --
	(141.11, 66.38) --
	(139.14, 64.02) --
	(137.18, 63.30) --
	(135.22, 62.46) --
	(133.26, 61.81) --
	(131.30, 61.03) --
	(129.33, 60.45) --
	(127.37, 61.10) --
	(125.41, 59.33) --
	(123.45, 58.57) --
	(121.49, 58.11) --
	(119.52, 57.44) --
	(117.56, 57.12) --
	(115.60, 56.37) --
	(113.64, 55.91) --
	(111.67, 55.41) --
	(109.71, 55.03) --
	(107.75, 54.47) --
	(105.79, 54.06) --
	(103.83, 53.82) --
	(101.86, 53.32) --
	( 99.90, 52.95) --
	( 97.94, 52.75) --
	( 95.98, 52.18) --
	( 94.02, 51.82) --
	( 92.05, 51.62) --
	( 90.09, 51.28) --
	( 88.13, 50.96) --
	( 86.17, 50.76) --
	( 84.21, 50.51) --
	( 82.24, 50.27) --
	( 80.28, 50.07) --
	( 78.32, 49.88) --
	( 76.36, 49.77) --
	( 74.40, 49.53) --
	( 72.43, 49.35) --
	( 70.47, 49.21) --
	( 68.51, 49.08) --
	( 66.55, 49.32) --
	( 64.59, 48.88) --
	( 62.62, 48.80) --
	( 60.66, 48.73) --
	( 58.70, 48.66) --
	( 56.74, 48.63) --
	( 54.78, 48.58) --
	( 52.81, 48.54) --
	( 50.85, 48.52) --
	( 48.89, 48.49) --
	cycle;
\definecolor{drawColor}{gray}{0.20}

\path[draw=drawColor,line width= 0.6pt,line join=round] ( 48.89, 48.49) --
	( 50.85, 48.52) --
	( 52.81, 48.55) --
	( 54.78, 48.58) --
	( 56.74, 48.63) --
	( 58.70, 48.67) --
	( 60.66, 48.74) --
	( 62.62, 48.80) --
	( 64.59, 48.90) --
	( 66.55, 49.33) --
	( 68.51, 49.09) --
	( 70.47, 49.27) --
	( 72.43, 49.36) --
	( 74.40, 49.60) --
	( 76.36, 49.88) --
	( 78.32, 50.03) --
	( 80.28, 50.20) --
	( 82.24, 50.43) --
	( 84.21, 50.69) --
	( 86.17, 50.92) --
	( 88.13, 51.15) --
	( 90.09, 51.46) --
	( 92.05, 51.72) --
	( 94.02, 52.01) --
	( 95.98, 52.33) --
	( 97.94, 53.11) --
	( 99.90, 53.07) --
	(101.86, 53.44) --
	(103.83, 53.99) --
	(105.79, 54.22) --
	(107.75, 54.65) --
	(109.71, 55.23) --
	(111.67, 55.57) --
	(113.64, 56.13) --
	(115.60, 56.56) --
	(117.56, 57.35) --
	(119.52, 57.66) --
	(121.49, 58.33) --
	(123.45, 58.78) --
	(125.41, 59.56) --
	(127.37, 61.34) --
	(129.33, 60.81) --
	(131.30, 61.30) --
	(133.26, 62.06) --
	(135.22, 62.79) --
	(137.18, 63.57) --
	(139.14, 64.33) --
	(141.11, 67.09) --
	(143.07, 65.94) --
	(145.03, 67.05) --
	(146.99, 67.60) --
	(148.95, 68.56) --
	(150.92, 69.64) --
	(152.88, 70.51) --
	(154.84, 71.48) --
	(156.80, 72.57) --
	(158.76, 73.59) --
	(160.73, 74.21) --
	(162.69, 75.69) --
	(164.65, 76.63) --
	(166.61, 77.98) --
	(168.57, 78.93) --
	(170.54, 80.61) --
	(172.50, 82.04) --
	(174.46, 83.02) --
	(176.42, 83.85) --
	(178.38, 85.42) --
	(180.35, 86.56) --
	(182.31, 88.40) --
	(184.27, 89.19) --
	(186.23, 91.05) --
	(188.19, 92.17) --
	(190.16, 93.94) --
	(192.12, 95.46) --
	(194.08, 97.28) --
	(196.04, 98.79) --
	(198.00,100.49) --
	(199.97,101.33) --
	(201.93,104.03) --
	(203.89,105.84) --
	(205.85,107.53) --
	(207.81,110.14) --
	(209.78,110.05) --
	(211.74,112.20) --
	(213.70,115.04) --
	(215.66,116.04) --
	(217.62,118.74) --
	(219.59,120.97) --
	(221.55,122.64) --
	(223.51,124.57) --
	(225.47,126.59) --
	(227.43,129.33) --
	(229.40,130.79) --
	(231.36,133.54) --
	(233.32,136.08) --
	(235.28,139.15) --
	(237.25,140.14) --
	(239.21,143.55) --
	(241.17,146.01) --
	(243.13,148.69) --
	(245.09,151.61) --
	(247.06,152.17) --
	(249.02,156.59) --
	(250.98,160.84) --
	(252.94,162.06) --
	(254.90,163.89) --
	(256.87,168.20) --
	(258.83,170.28) --
	(260.79,174.28) --
	(262.75,174.80) --
	(264.71,178.36) --
	(266.68,182.63) --
	(268.64,185.33) --
	(270.60,188.81) --
	(272.56,190.43) --
	(274.52,193.87) --
	(276.49,196.29) --
	(278.45,199.86) --
	(280.41,203.41) --
	(282.37,210.12) --
	(284.33,212.19) --
	(286.30,213.96) --
	(288.26,218.13) --
	(290.22,220.55) --
	(292.18,223.79) --
	(294.14,223.78) --
	(296.11,231.29) --
	(298.07,250.75) --
	(300.03,239.15) --
	(301.99,241.95) --
	(303.95,247.86) --
	(305.92,252.65) --
	(307.88,254.34) --
	(309.84,258.32) --
	(311.80,263.82) --
	(313.76,269.70) --
	(315.73,273.05) --
	(317.69,273.72) --
	(319.65,279.82) --
	(321.61,283.91) --
	(323.57,288.42) --
	(325.54,294.70) --
	(327.50,297.12) --
	(329.46,303.55) --
	(331.42,305.76) --
	(333.38,312.04) --
	(335.35,318.22) --
	(337.31,321.40) --
	(339.27,325.33) --
	(341.23,332.63);
\end{scope}
\begin{scope}
\definecolor{drawColor}{RGB}{190,190,190}

\path[draw=drawColor,line width= 0.6pt,line join=round] ( 34.27, 33.88) --
	( 34.27,355.45);
\end{scope}
\begin{scope}
\definecolor{drawColor}{gray}{0.30}

\node[text=drawColor,anchor=base east,inner sep=0pt, outer sep=0pt, scale=  0.88] at ( 29.32, 45.42) {0};

\node[text=drawColor,anchor=base east,inner sep=0pt, outer sep=0pt, scale=  0.88] at ( 29.32,109.12) {10};

\node[text=drawColor,anchor=base east,inner sep=0pt, outer sep=0pt, scale=  0.88] at ( 29.32,172.82) {20};

\node[text=drawColor,anchor=base east,inner sep=0pt, outer sep=0pt, scale=  0.88] at ( 29.32,236.51) {30};

\node[text=drawColor,anchor=base east,inner sep=0pt, outer sep=0pt, scale=  0.88] at ( 29.32,300.21) {40};
\end{scope}
\begin{scope}
\definecolor{drawColor}{gray}{0.20}

\path[draw=drawColor,line width= 0.6pt,line join=round] ( 31.52, 48.45) --
	( 34.27, 48.45);

\path[draw=drawColor,line width= 0.6pt,line join=round] ( 31.52,112.15) --
	( 34.27,112.15);

\path[draw=drawColor,line width= 0.6pt,line join=round] ( 31.52,175.85) --
	( 34.27,175.85);

\path[draw=drawColor,line width= 0.6pt,line join=round] ( 31.52,239.54) --
	( 34.27,239.54);

\path[draw=drawColor,line width= 0.6pt,line join=round] ( 31.52,303.24) --
	( 34.27,303.24);
\end{scope}
\begin{scope}
\definecolor{drawColor}{RGB}{190,190,190}

\path[draw=drawColor,line width= 0.6pt,line join=round] ( 34.27, 33.88) --
	(355.85, 33.88);
\end{scope}
\begin{scope}
\definecolor{drawColor}{gray}{0.20}

\path[draw=drawColor,line width= 0.6pt,line join=round] ( 46.93, 31.13) --
	( 46.93, 33.88);

\path[draw=drawColor,line width= 0.6pt,line join=round] (145.03, 31.13) --
	(145.03, 33.88);

\path[draw=drawColor,line width= 0.6pt,line join=round] (243.13, 31.13) --
	(243.13, 33.88);

\path[draw=drawColor,line width= 0.6pt,line join=round] (341.23, 31.13) --
	(341.23, 33.88);
\end{scope}
\begin{scope}
\definecolor{drawColor}{gray}{0.30}

\node[text=drawColor,anchor=base,inner sep=0pt, outer sep=0pt, scale=  0.88] at ( 46.93, 22.87) {0};

\node[text=drawColor,anchor=base,inner sep=0pt, outer sep=0pt, scale=  0.88] at (145.03, 22.87) {50};

\node[text=drawColor,anchor=base,inner sep=0pt, outer sep=0pt, scale=  0.88] at (243.13, 22.87) {100};

\node[text=drawColor,anchor=base,inner sep=0pt, outer sep=0pt, scale=  0.88] at (341.23, 22.87) {150};
\end{scope}
\begin{scope}
\definecolor{drawColor}{RGB}{0,0,0}

\node[text=drawColor,anchor=base,inner sep=0pt, outer sep=0pt, scale=  1.50] at (195.06,  5.51) {$n$};
\end{scope}
\begin{scope}
\definecolor{drawColor}{RGB}{0,0,0}

\node[text=drawColor,rotate= 90.00,anchor=base,inner sep=0pt, outer sep=0pt, scale=  1.50] at ( 17.19,194.66) {Execution time in ms};
\end{scope}
\end{tikzpicture}

%% file: 5-faithful-noe.tex
\section{Faithfully Rounded Evaluation of Noe's Recursion}\label{sec:noe_faithful}

We implemented faithfully rounded\footnote{That is, the result is either exact (if the exact value is a floating point number) or it is one of the two closest floating point numbers.} floating-point computations as described by \cite{rump2017} as a portable single-header C++11 library.\footnote{Available at \url{https://github.com/jvschroeder/PairArithmetic/}.} Utilizing this library we implemented the generalization of Noe's recursion presented in Section \ref{sec33} obtaining faithfully rounded results if no underflow occurs.\footnote{In our experience this is usually the case if the values of $\Psi$ are not too close to the smallest (in absolute value) normal double, which equals $2^{-1022}\approx 2.225\cdot 10^{-308}$ on most computer architectures.} In case of an underflow the results are smaller than the true values of $\Psi$, but never less than zero. Parallelization was implemented using \href{https://www.threadingbuildingblocks.org/}{Intel\textregistered Threading Building Blocks (TBB)}.

Notice that Noe's recursion (and our generalization thereof) satisfies the NIC principle (\textit{No Inaccurate Cancellation}, cf. \cite[Definition 2.2]{rump2017}), that is there are no sums where at least one summand is not an input to the algorithm and the summands have opposite signs. Thus, by examining the evaluation tree (cf. \cite[Definition 2.2]{rump2017}) of a concrete implementation, it is possible to
calculate a number $k = k(n_1, n_2)$ according to Equation (11) of \cite{rump2017}. The result will be faithfully rounded if  no under- or overflow occurs, and $k\leq 2^{26}-2$ (when utilising a double precision floating point numbers). For our concrete implementation we obtain $k(n_1,n_2) =n_1\cdot n_2+8\cdot(n_1+n_2)-7$, provided that $n_1+n_2\geq 2$. Thus (assuming that no over- or underflow occurs) the result is guaranteed to be faithfully rounded if $n_1,n_2\leq 8184$. Our implementation could be, in terms of $k$, significantly improved by using binary summation. For example, for $n_1=n_2=400$ we obtain $k(400,400)=166{,}398$, while the corresponding number of $k$ in the case of binary summation would equal $17{,}421$. The latter improvement however comes at an additional computational cost, and may be considered mostly of theoretical interest since the calculation for $n_1=n_2=400$ already takes approximately $27$ minutes on a $4$ core Intel CPU. Figure \ref{fig:comp_bolshev_noe} compares the runtime of our implementation of our generalization of Noe's recursion to that of the algorithm from the previous section. It becomes apparent that Noe's recursion with faithful rounding is much faster then Bolshev's recursion implemented in rational arithmetic. For practical applications, we therefore recommend Noe's recursion with faithful rounding, at least if a fixed numerical precision is sufficient.

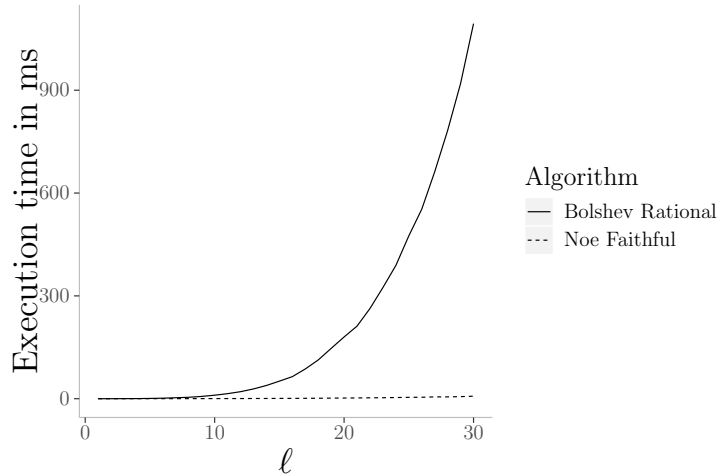
\begin{figure}[!htb]
	\centering
    \resizebox{.7\linewidth}{!}{\input{./comp_bolshev_noe}}
	\caption{Comparison of the runtime of Algorithm \ref{alg:bolshev_general} (where $n_1 = n_2 = \ell$) implemented in rational arithmetic with Noe's recursion implemented in faithfully rounded floating point arithmetic.}
	\label{fig:comp_bolshev_noe}
\end{figure}

%% file: comp_bolshev_noe.tex
\begin{tikzpicture}[x=1pt,y=1pt]
\definecolor{fillColor}{RGB}{255,255,255}
\begin{scope}
\definecolor{drawColor}{RGB}{255,255,255}
\definecolor{fillColor}{RGB}{255,255,255}

\path[draw=drawColor,line width= 0.6pt,line join=round,line cap=round,fill=fillColor] (  0.00, 56.47) rectangle (361.35,304.88);
\end{scope}
\begin{scope}
\definecolor{drawColor}{RGB}{255,255,255}

\path[draw=drawColor,line width= 0.3pt,line join=round] ( 38.67,125.56) --
	(248.11,125.56);

\path[draw=drawColor,line width= 0.3pt,line join=round] ( 38.67,177.76) --
	(248.11,177.76);

\path[draw=drawColor,line width= 0.3pt,line join=round] ( 38.67,229.96) --
	(248.11,229.96);

\path[draw=drawColor,line width= 0.3pt,line join=round] ( 38.67,282.16) --
	(248.11,282.16);

\path[draw=drawColor,line width= 0.3pt,line join=round] ( 74.45, 89.95) --
	( 74.45,299.38);

\path[draw=drawColor,line width= 0.3pt,line join=round] (140.11, 89.95) --
	(140.11,299.38);

\path[draw=drawColor,line width= 0.3pt,line join=round] (205.76, 89.95) --
	(205.76,299.38);

\path[draw=drawColor,line width= 0.6pt,line join=round] ( 38.67, 99.47) --
	(248.11, 99.47);

\path[draw=drawColor,line width= 0.6pt,line join=round] ( 38.67,151.66) --
	(248.11,151.66);

\path[draw=drawColor,line width= 0.6pt,line join=round] ( 38.67,203.86) --
	(248.11,203.86);

\path[draw=drawColor,line width= 0.6pt,line join=round] ( 38.67,256.06) --
	(248.11,256.06);

\path[draw=drawColor,line width= 0.6pt,line join=round] ( 41.63, 89.95) --
	( 41.63,299.38);

\path[draw=drawColor,line width= 0.6pt,line join=round] (107.28, 89.95) --
	(107.28,299.38);

\path[draw=drawColor,line width= 0.6pt,line join=round] (172.93, 89.95) --
	(172.93,299.38);

\path[draw=drawColor,line width= 0.6pt,line join=round] (238.59, 89.95) --
	(238.59,299.38);
\definecolor{drawColor}{RGB}{0,0,0}

\path[draw=drawColor,line width= 0.6pt,line join=round] ( 48.19, 99.47) --
	( 54.76, 99.47) --
	( 61.32, 99.50) --
	( 67.89, 99.53) --
	( 74.45, 99.61) --
	( 81.02, 99.73) --
	( 87.58, 99.96) --
	( 94.15,100.19) --
	(100.71,100.63) --
	(107.28,101.25) --
	(113.85,102.03) --
	(120.41,103.03) --
	(126.98,104.43) --
	(133.54,106.18) --
	(140.11,108.39) --
	(146.67,110.66) --
	(153.24,114.56) --
	(159.80,119.13) --
	(166.37,125.01) --
	(172.93,130.81) --
	(179.50,136.32) --
	(186.06,145.27) --
	(192.63,155.91) --
	(199.20,167.13) --
	(205.76,182.12) --
	(212.33,195.69) --
	(218.89,214.66) --
	(225.46,235.66) --
	(232.02,259.51) --
	(238.59,289.86);

\path[draw=drawColor,line width= 0.6pt,dash pattern=on 2pt off 2pt ,line join=round] ( 48.19, 99.47) --
	( 54.76, 99.47) --
	( 61.32, 99.48) --
	( 67.89, 99.48) --
	( 74.45, 99.48) --
	( 81.02, 99.49) --
	( 87.58, 99.50) --
	( 94.15, 99.51) --
	(100.71, 99.52) --
	(107.28, 99.53) --
	(113.85, 99.54) --
	(120.41, 99.55) --
	(126.98, 99.58) --
	(133.54, 99.59) --
	(140.11, 99.62) --
	(146.67, 99.64) --
	(153.24, 99.67) --
	(159.80, 99.69) --
	(166.37, 99.77) --
	(172.93, 99.80) --
	(179.50, 99.86) --
	(186.06, 99.93) --
	(192.63, 99.96) --
	(199.20,100.06) --
	(205.76,100.17) --
	(212.33,100.21) --
	(218.89,100.39) --
	(225.46,100.39) --
	(232.02,100.52) --
	(238.59,100.74);
\end{scope}
\begin{scope}
\definecolor{drawColor}{RGB}{190,190,190}

\path[draw=drawColor,line width= 0.6pt,line join=round] ( 38.67, 89.95) --
	( 38.67,299.38);
\end{scope}
\begin{scope}
\definecolor{drawColor}{gray}{0.30}

\node[text=drawColor,anchor=base east,inner sep=0pt, outer sep=0pt, scale=  0.88] at ( 33.72, 96.43) {0};

\node[text=drawColor,anchor=base east,inner sep=0pt, outer sep=0pt, scale=  0.88] at ( 33.72,148.63) {300};

\node[text=drawColor,anchor=base east,inner sep=0pt, outer sep=0pt, scale=  0.88] at ( 33.72,200.83) {600};

\node[text=drawColor,anchor=base east,inner sep=0pt, outer sep=0pt, scale=  0.88] at ( 33.72,253.03) {900};
\end{scope}
\begin{scope}
\definecolor{drawColor}{gray}{0.20}

\path[draw=drawColor,line width= 0.6pt,line join=round] ( 35.92, 99.47) --
	( 38.67, 99.47);

\path[draw=drawColor,line width= 0.6pt,line join=round] ( 35.92,151.66) --
	( 38.67,151.66);

\path[draw=drawColor,line width= 0.6pt,line join=round] ( 35.92,203.86) --
	( 38.67,203.86);

\path[draw=drawColor,line width= 0.6pt,line join=round] ( 35.92,256.06) --
	( 38.67,256.06);
\end{scope}
\begin{scope}
\definecolor{drawColor}{RGB}{190,190,190}

\path[draw=drawColor,line width= 0.6pt,line join=round] ( 38.67, 89.95) --
	(248.11, 89.95);
\end{scope}
\begin{scope}
\definecolor{drawColor}{gray}{0.20}

\path[draw=drawColor,line width= 0.6pt,line join=round] ( 41.63, 87.20) --
	( 41.63, 89.95);

\path[draw=drawColor,line width= 0.6pt,line join=round] (107.28, 87.20) --
	(107.28, 89.95);

\path[draw=drawColor,line width= 0.6pt,line join=round] (172.93, 87.20) --
	(172.93, 89.95);

\path[draw=drawColor,line width= 0.6pt,line join=round] (238.59, 87.20) --
	(238.59, 89.95);
\end{scope}
\begin{scope}
\definecolor{drawColor}{gray}{0.30}

\node[text=drawColor,anchor=base,inner sep=0pt, outer sep=0pt, scale=  0.88] at ( 41.63, 78.94) {0};

\node[text=drawColor,anchor=base,inner sep=0pt, outer sep=0pt, scale=  0.88] at (107.28, 78.94) {10};

\node[text=drawColor,anchor=base,inner sep=0pt, outer sep=0pt, scale=  0.88] at (172.93, 78.94) {20};

\node[text=drawColor,anchor=base,inner sep=0pt, outer sep=0pt, scale=  0.88] at (238.59, 78.94) {30};
\end{scope}
\begin{scope}
\definecolor{drawColor}{RGB}{0,0,0}

\node[text=drawColor,anchor=base,inner sep=0pt, outer sep=0pt, scale=  1.50] at (143.39, 61.58) {$\ell$};
\end{scope}
\begin{scope}
\definecolor{drawColor}{RGB}{0,0,0}

\node[text=drawColor,rotate= 90.00,anchor=base,inner sep=0pt, outer sep=0pt, scale=  1.50] at ( 17.19,194.66) {Execution time in ms};
\end{scope}
\begin{scope}
\definecolor{fillColor}{RGB}{255,255,255}

\path[fill=fillColor] (259.11,167.20) rectangle (355.85,222.13);
\end{scope}
\begin{scope}
\definecolor{drawColor}{RGB}{0,0,0}

\node[text=drawColor,anchor=base west,inner sep=0pt, outer sep=0pt, scale=  1.10] at (264.61,208.08) {Algorithm};
\end{scope}
\begin{scope}
\definecolor{drawColor}{RGB}{255,255,255}
\definecolor{fillColor}{gray}{0.95}

\path[draw=drawColor,line width= 0.6pt,line join=round,line cap=round,fill=fillColor] (264.61,187.15) rectangle (279.06,201.61);
\end{scope}
\begin{scope}
\definecolor{drawColor}{RGB}{0,0,0}

\path[draw=drawColor,line width= 0.6pt,line join=round] (266.05,194.38) -- (277.62,194.38);
\end{scope}
\begin{scope}
\definecolor{drawColor}{RGB}{255,255,255}
\definecolor{fillColor}{gray}{0.95}

\path[draw=drawColor,line width= 0.6pt,line join=round,line cap=round,fill=fillColor] (264.61,172.70) rectangle (279.06,187.15);
\end{scope}
\begin{scope}
\definecolor{drawColor}{RGB}{0,0,0}

\path[draw=drawColor,line width= 0.6pt,dash pattern=on 2pt off 2pt ,line join=round] (266.05,179.93) -- (277.62,179.93);
\end{scope}
\begin{scope}
\definecolor{drawColor}{RGB}{0,0,0}

\node[text=drawColor,anchor=base west,inner sep=0pt, outer sep=0pt, scale=  0.88] at (284.56,191.35) {Bolshev Rational};
\end{scope}
\begin{scope}
\definecolor{drawColor}{RGB}{0,0,0}

\node[text=drawColor,anchor=base west,inner sep=0pt, outer sep=0pt, scale=  0.88] at (284.56,176.90) {Noe Faithful};
\end{scope}
\end{tikzpicture}

%% file: 6-applications.tex
\section{Applications in Multiple Hypothesis Testing} \label{sec6}

As discussed by \cite{roquain2011}, the values of $\Psi(i_1,i_2)$ for all $0\leq i_1\leq n_1$ and $0\leq i_2\leq n_2$ are important building blocks for calculating the joint distribution of the number of rejections $R$and the number of false rejections $V$ for step-up multiple tests. The random variables $V$ and $R$ play an important role when analyzing the type I and type II error behavior of such multiple tests. 
One important observation is, that the previously discussed recursions for calculating $\Psi(n_1,n_2)$ also calculate all such $\Psi(i_1,i_2)$'s as intermediate results.  

Following \cite{blanchard2014least} we consider $m\geq 2$ null hypotheses $H_1,\hdots, H_m$ which are simultaneously under consideration under one and the same statistical model. We assume that associated $p$-values $p_1,\hdots,p_m$ are available on which the multiple test operates. Furthermore, we assume that 
$p_1,\hdots,p_m$ (regarded as random variables) are jointly distributed according to one of the following models
\begin{description}
    \item[$\text{FM}(m,m_0,F)$] The $p_i$'s are stochastically independent with marginal distributions \begin{align*}
            p_i\sim\begin{cases}
                \text{Uni}[0,1]&\text{if }1\leq i\leq m_0,\\
                F&\text{if }m_0+1\leq i\leq m,
            \end{cases}
        \end{align*}
where $m_0$ denotes the number of true null hypotheses among 	$H_1,\hdots, H_m$ and $F$ is a given continuous cdf on $[0, 1]$.			
    \item[$\text{RM}(m,\pi_0,F)$] Let $M_0$ denote a binomially distributed random variable, \linebreak[4] 
		$M_0\sim\mathcal B(m,\pi_0)$. Conditionally on $M_0=m_0$, the $p_i$'s are jointly distributed according to $\text{FM}(m,m_0,F)$.
\end{description} 

A multiple test operating on $\mathbf{p} = (p_1, \hdots, p_m)^\top$ is a measurable mapping $\varphi:[0,1]^m\rightarrow\mathcal P([m])$, where hypothesis $H_i$ is rejected iff $i\in\varphi(\mathbf{p})$. Under $\text{FM}(m,m_0,F)$ denote by $M_0 \equiv m_0$ a constant random variable. Then the (random) number of rejections of the multiple test $\varphi$ is given by $R(\varphi, \mathbf{p}):=\left|\varphi(\mathbf{p})\right|$, and $V(\varphi, \mathbf{p}):=\left|\varphi(\mathbf{p})\cap[M_0]\right|$ is the (random) number of false rejections (type I errors).

In the following we will consider step-up procedures $\varphi=\text{SU}_\mathbf{t}$ with critical values 
$\mathbf{t} = (t_1, \hdots, t_m)^\top \in(0,1)^m$ such that $t_1\leq \hdots \leq t_m$. The corresponding decision rule can be written as
\begin{align*}
\text{SU}_\mathbf{t}(\mathbf{p})&:=\left[\max\left(\left\{0\right\}\cup\left\{i\in[m]: 
p_{i:m}\leq t_i\right\}\right)\right], \text{~~where~~} [0] := \emptyset.
\end{align*}

Summarizing results of \cite{roquain2011}, the joint distribution of $R$ and $V$ for any step-up procedure $\text{SU}_\mathbf{t}$ has the following properties.

\begin{lemma}\label{lemma:exact_dist}
Let $0 \leq j\leq k\leq m$.
\begin{enumerate}
\item[(i)]
    Under the unconditional model $RM(m,\pi_0,F)$ it holds that
    \begin{align*}
        &\mathbb{P}_{m, \pi_0,F}\left(V(\text{SU}_\mathbf{t}, \mathbf{p}) =j,R(\text{SU}_\mathbf{t}, \mathbf{p})=k\right)\\
        &=\binom{m}{k}\binom{k}{j}\tilde\pi_0^j(1-\tilde\pi_0)^{k-j}G(t_k)^k\Psi_{m-k,0}^{\text{Uni}[0,1], F}(1-G(t_m),\hdots,1-G(t_{k+1}))
    \end{align*} where $\tilde\pi_0:=\pi_0t_k/G(t_k)$ and $G(t):=\pi_0t+(1-\pi_0)F(t)$.
    
 \item[(ii)] Under the conditional model $FM(m,m_0,F)$ it holds that\begin{align*}
        &\mathbb{P}_{m,m_0,F}\left(V(\text{SU}_\mathbf{t}, \mathbf{p})=j,R(\text{SU}_\mathbf{t}, \mathbf{p})=k\right)\\
        &=\binom{m_0}{j}\binom{m-m_0}{k-j}t_k^j\left( F(t_k)\right)^{k-j}\Psi^{\text{Uni}[0,1],\bar F}_{\substack{m-k-(m_0-j),\\ m_0-j}}(1-t_m,\hdots,1-t_{k+1}),
    \end{align*} where $0\leq j\leq m_0$ and $\bar{F}(t) := 1 - F(t)$.
		\end{enumerate}
\end{lemma}

Combining Lemma \ref{lemma:exact_dist} with the previously discussed efficient evaluation of $\Psi$ it is possible to calculate various summary statistics pertaining to the joint distribution of $(V,R,M_0)$ under the above models.

\begin{definition} $ $
\begin{itemize}
    \item[(a)] The FDR of $\text{SU}_\mathbf{t}(\mathbf{p})$ is given by the expectation of the false discovery proportion (FDP) of $\text{SU}_\mathbf{t}(\mathbf{p})$, which is given by 
		\begin{equation*}
     \operatorname{FDP}(\text{SU}_\mathbf{t}, \mathbf{p}) := \frac{V(\text{SU}_\mathbf{t},\mathbf{p})}{R(\text{SU}_\mathbf{t}, \mathbf{p})\imax 1}.
		\end{equation*}
    \item[(b)] 
		Considering the number of correct rejections $R(\text{SU}_\mathbf{t},\mathbf{p})-V(\text{SU}_\mathbf{t},\mathbf{p})$ the average power of $\text{SU}_\mathbf{t}(\mathbf{p})$ is given by
		\begin{equation}
            \operatorname{Pow}_\text{avg}(\text{SU}_\mathbf{t}) :=\mathbb{E}\left[\frac{R(\text{SU}_\mathbf{t},\mathbf{p})-V(\text{SU}_\mathbf{t},\mathbf{p})}{m-M_0}\right],\label{eq:avg_power}
        \end{equation}
				where the convention $\frac{0}{0}=0$ is utilized and where $\mathbb E =\mathbb E_{m,\pi_0,F}$ (under $RM(m,\pi_0,F)$) or $\mathbb E =\mathbb E_{m,m_0,F}$ (under $FM(m,m_0,F)$), respectively.
		\item[(c)] The $\lambda-$power is the probability of rejecting at least $\lambda\cdot(m-M_0)$ of the false hypotheses:\begin{equation}
            \operatorname{Pow}_\lambda(\text{SU}_\mathbf{t}):=\mathbb{P}\left(\frac{R(\text{SU}_\mathbf{t},\mathbf{p})-V(\text{SU}_\mathbf{t},\mathbf{p})}{m-M_0}\geq\lambda\right)\label{eq:lambda_power}
        \end{equation} where, again, the convention $\frac{0}{0}=0$ is utilised and where $\mathbb P =\mathbb P_{m,\pi_0,F}$ (under $RM(m,\pi_0,F)$) or $\mathbb P =\mathbb P_{m,m_0,F}$ (under $FM(m,m_0,F)$), respectively.
\end{itemize}
\end{definition}				
				
In order to provide some numerical illustrations, we first consider the average power (cf.\ \eqref{eq:avg_power}) under $FM(m,m_0,F)$ where \begin{equation}
F(t):=1+\Phi\left(\Phi^{-1}\left(\frac{t}{2}\right)-\sqrt{N}\right)-\Phi\left(\Phi^{-1}\left(1-\frac{t}{2}\right)-\sqrt{N}\right)
\label{eq:model_glueck}
\end{equation} and $N=5$. This is the setting considered in \cite[Table 2]{glueck2008exact} where the average power for $m\leq 5$ was calculated for the Benjamini-Hochberg procedure (controlling the FDR at $\alpha=0.05$) for $m$ independent two-sided one sample z-tests. In our notation, the Benjamini-Hochberg (linear step-up test) procedure equals $\text{SU}_\mathbf{t}$ with $t_i = i \alpha / m$ for $i \in [m]$. Table \ref{table:average_power_glueck} illustrates the results obtained for $m,m_0\leq 50$. Due to space constraints only the first six columns  (corresponding to $m_0\leq 5$) are presented. The calculation of the full table (not presented here) took less than a second for an $m$ one magnitude larger (50 instead of 5) than the one considered by \cite{glueck2008exact}. Figure \ref{fig:avg_power_time} illustrates the time needed to calculate one row of such a table corresponding to some $m\in\mathbb{N}$ when utilizing our proposed algorithms.

As a second example, we consider the computation of $\operatorname{Pow}_\lambda(\text{SU}_\mathbf{t})$ from \eqref{eq:lambda_power}. Again, we choose $\mathbf{t}$ as in the Benjamini-Hochberg case. An asymptotic approximation of this quantity for
\[F(t):=F_{\nu,\mu}\left(F_{\nu,0}^{-1}\left(\frac{t}{2}\right)\right)-F_{\nu,\mu}\left(-F_{\nu,0}^{-1}\left(\frac{t}{2}\right)\right),
\]
where $F_{\nu,\mu}$ denotes the distribution function of a non-central chi-squared random variable with $\nu$ degrees of freedom and non-centrality parameter $\mu$, is given in \cite[Table 3]{izmirlian2018average}. Our results can be used to calculate $\operatorname{Pow}_\lambda(\text{SU}_\mathbf{t})$. Table \ref{table:lambda_power} gives the faithfully rounded values for the $0.9$-power for the parameters considered in \cite[Table 3]{izmirlian2018average}.

We conclude by giving an example for the exact distribution of the FDP which shows why the FDR is not always an appropriate summary statistic. Consider again the multiple two-sided z-test described in \cite{glueck2008exact}, that is $F$ given by \eqref{eq:model_glueck}, for $N=m=50$ and $m_0=5$. It is clear that in Figure \ref{fig:fdp_dist} the distribution of the FDP is neither symmetric about its mean (the FDR, which is depicted as dotted vertical line) nor concentrated around the FDR. A similar argumentation has been used by, among others, \cite{blanchard2014least} and \cite{Delattre-Roquain-Romano-Wolf} in order to motivate the computation of the full distribution of the FDP and to control its quantiles. The latter task is inherently computationally demanding.

\include{./average_power_b}

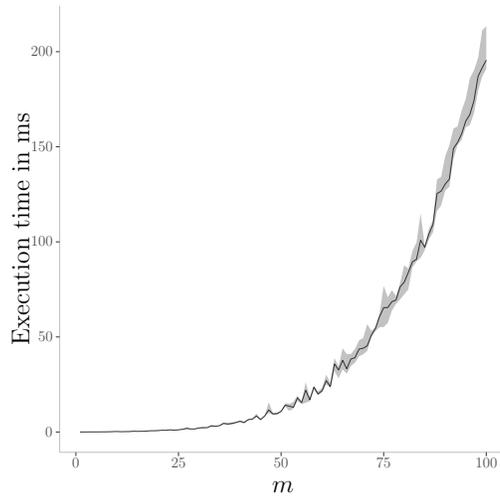
\begin{figure}[!htb]
	\centering
    \resizebox{.5\linewidth}{!}{\input{./average_power_time}}
	\caption{Time needed to calculate the average power of the Benjamini-Hochberg procedure for $m$ hypotheses.}
	\label{fig:avg_power_time}
\end{figure}

\include{./lambda_power}

\begin{figure}[!htb]
	\centering
    \input{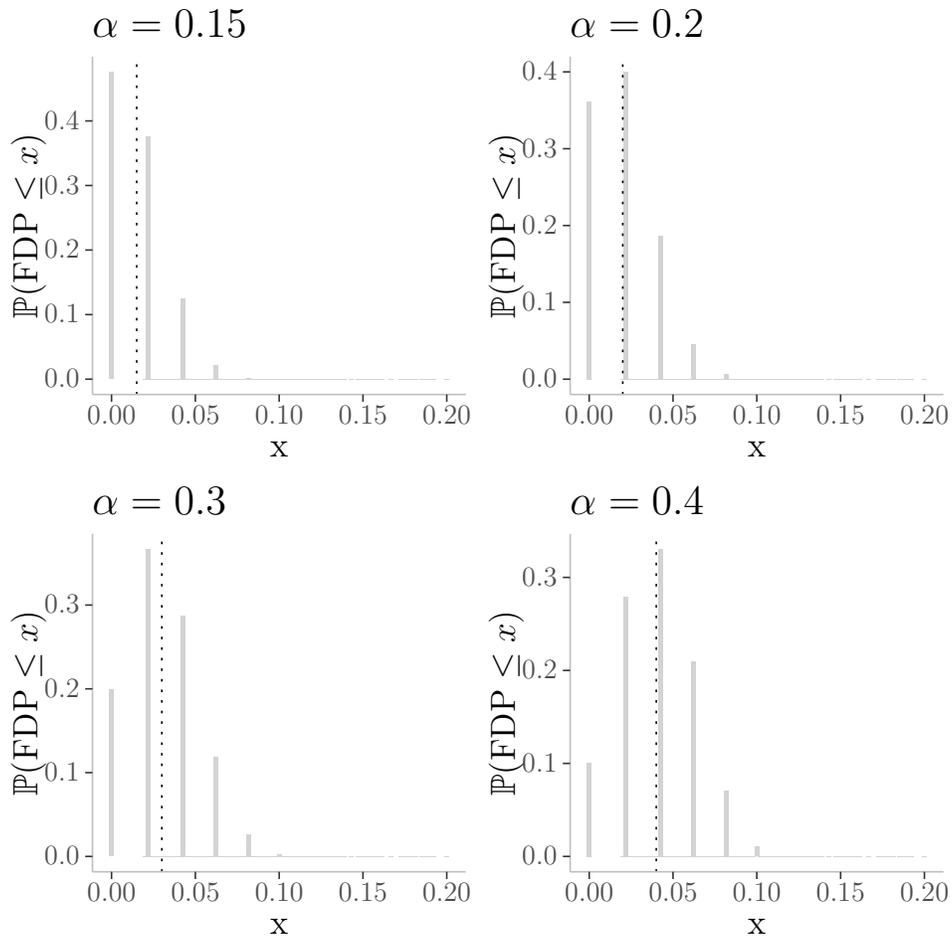}
	\caption{Distribution of the FDP for the Benjamini-Hochberg procedure (controlling the FDR at $\alpha$) for $m$ independent two-sided one sample z-tests (sample size $N=5$, common variance $\sigma^2=1$, $\mu_0=0,\mu_1=1$, cf.  \eqref{eq:model_glueck}). The dotted vertical line is the actual FDR of the test.}
	\label{fig:fdp_dist}
\end{figure}

%% file: average_power_b.tex
\begin{table}[ht]
\centering
\renewrobustcmd{\bfseries}{\fontseries{b}\selectfont}
\begin{tabular}{lrrrrrr}
  \hline
 $m$& $m_0=0$ & $m_0=1$ & $m_0=2$ & $m_0=3$ & $m_0=4$ & $m_0=5$ \\
  \hline
2 & \textbf{0.56539} & \textbf{0.50342} &  &  &  &  \\
  3 & \textbf{0.54576} & \textbf{0.49842} & \textbf{0.44439} &  &  &  \\
  4 & \textbf{0.53446} & \textbf{0.49583} & \textbf{0.45256} & \textbf{0.40451} &  &  \\
  5 & \textbf{0.52712} & \textbf{0.49440} & \textbf{0.45819} & \textbf{0.41837} & \textbf{0.37494} &  \\
  6 & 0.52201 & 0.49357 & 0.46241 & 0.42840 & 0.39148 & 0.35175 \\
  7 & 0.51827 & 0.49310 & 0.46574 & 0.43606 & 0.40399 & 0.36955 \\
  8 & 0.51543 & 0.49285 & 0.46846 & 0.44214 & 0.41383 & 0.38349 \\
  9 & 0.51320 & 0.49273 & 0.47073 & 0.44710 & 0.42178 & 0.39473 \\
  10 & 0.51142 & 0.49270 & 0.47266 & 0.45124 & 0.42836 & 0.40398 \\
  11 & 0.50997 & 0.49272 & 0.47434 & 0.45475 & 0.43390 & 0.41174 \\
  12 & 0.50877 & 0.49278 & 0.47580 & 0.45777 & 0.43863 & 0.41833 \\
  13 & 0.50776 & 0.49286 & 0.47709 & 0.46039 & 0.44271 & 0.42401 \\
  14 & 0.50690 & 0.49295 & 0.47823 & 0.46268 & 0.44627 & 0.42895 \\
  15 & 0.50616 & 0.49306 & 0.47925 & 0.46472 & 0.44941 & 0.43329 \\
  16 & 0.50552 & 0.49316 & 0.48018 & 0.46653 & 0.45219 & 0.43712 \\
  17 & 0.50497 & 0.49327 & 0.48101 & 0.46815 & 0.45467 & 0.44053 \\
  18 & 0.50447 & 0.49338 & 0.48177 & 0.46962 & 0.45690 & 0.44358 \\
  19 & 0.50404 & 0.49348 & 0.48246 & 0.47094 & 0.45891 & 0.44634 \\
  20 & 0.50365 & 0.49358 & 0.48309 & 0.47215 & 0.46074 & 0.44883 \\
  21 & 0.50330 & 0.49368 & 0.48367 & 0.47325 & 0.46240 & 0.45109 \\
  22 & 0.50298 & 0.49378 & 0.48421 & 0.47427 & 0.46392 & 0.45316 \\
  23 & 0.50269 & 0.49387 & 0.48471 & 0.47520 & 0.46532 & 0.45505 \\
  24 & 0.50243 & 0.49395 & 0.48517 & 0.47605 & 0.46660 & 0.45679 \\
  25 & 0.50219 & 0.49404 & 0.48559 & 0.47685 & 0.46779 & 0.45840 \\
  26 & 0.50198 & 0.49412 & 0.48599 & 0.47759 & 0.46889 & 0.45988 \\
  27 & 0.50178 & 0.49420 & 0.48637 & 0.47828 & 0.46991 & 0.46126 \\
  28 & 0.50159 & 0.49427 & 0.48671 & 0.47892 & 0.47086 & 0.46254 \\
  29 & 0.50142 & 0.49434 & 0.48704 & 0.47952 & 0.47175 & 0.46373 \\
  30 & 0.50126 & 0.49441 & 0.48735 & 0.48008 & 0.47258 & 0.46485 \\
  31 & 0.50111 & 0.49447 & 0.48764 & 0.48060 & 0.47336 & 0.46589 \\
  32 & 0.50097 & 0.49453 & 0.48791 & 0.48110 & 0.47409 & 0.46687 \\
  33 & 0.50084 & 0.49459 & 0.48817 & 0.48157 & 0.47478 & 0.46779 \\
  34 & 0.50072 & 0.49465 & 0.48841 & 0.48201 & 0.47542 & 0.46866 \\
  35 & 0.50060 & 0.49470 & 0.48864 & 0.48242 & 0.47604 & 0.46948 \\
  36 & 0.50050 & 0.49475 & 0.48886 & 0.48282 & 0.47661 & 0.47025 \\
  37 & 0.50040 & 0.49480 & 0.48907 & 0.48319 & 0.47716 & 0.47098 \\
  38 & 0.50030 & 0.49485 & 0.48926 & 0.48354 & 0.47768 & 0.47167 \\
  39 & 0.50021 & 0.49489 & 0.48945 & 0.48388 & 0.47817 & 0.47233 \\
  40 & 0.50012 & 0.49494 & 0.48963 & 0.48420 & 0.47864 & 0.47295 \\
  41 & 0.50004 & 0.49498 & 0.48980 & 0.48451 & 0.47909 & 0.47354 \\
  42 & 0.49996 & 0.49502 & 0.48996 & 0.48480 & 0.47951 & 0.47411 \\
  43 & 0.49989 & 0.49506 & 0.49012 & 0.48508 & 0.47992 & 0.47465 \\
  44 & 0.49982 & 0.49509 & 0.49027 & 0.48534 & 0.48031 & 0.47516 \\
  45 & 0.49975 & 0.49513 & 0.49041 & 0.48559 & 0.48068 & 0.47565 \\
  46 & 0.49969 & 0.49516 & 0.49055 & 0.48584 & 0.48103 & 0.47612 \\
  47 & 0.49963 & 0.49520 & 0.49068 & 0.48607 & 0.48137 & 0.47657 \\
  48 & 0.49957 & 0.49523 & 0.49081 & 0.48629 & 0.48169 & 0.47700 \\
  49 & 0.49951 & 0.49526 & 0.49093 & 0.48651 & 0.48201 & 0.47741 \\
  50 & 0.49946 & 0.49529 & 0.49104 & 0.48672 & 0.48230 & 0.47781 \\
   \hline
\end{tabular}
\caption{Average Power of the Benjamini-Hochberg procedure (controlling the FDR at $\alpha=0.05$) for $m$ independent two-sided one sample z-tests (sample size $N=5$, common variance $\sigma^2=1$, $\mu_0=0,\mu_1=1$, cf.  \eqref{eq:model_glueck}) when $m_0$ hypotheses are true. The bold values are exactly those in \cite[Table 2]{glueck2008exact}.}
\label{table:average_power_glueck}
\end{table}

%% file: average_power_time.tex
\begin{tikzpicture}[x=1pt,y=1pt]
\definecolor{fillColor}{RGB}{255,255,255}
\begin{scope}
\definecolor{drawColor}{RGB}{255,255,255}
\definecolor{fillColor}{RGB}{255,255,255}

\path[draw=drawColor,line width= 0.6pt,line join=round,line cap=round,fill=fillColor] (  0.00,  2.60) rectangle (361.35,358.75);
\end{scope}
\begin{scope}
\definecolor{drawColor}{RGB}{255,255,255}

\path[draw=drawColor,line width= 0.3pt,line join=round] ( 38.67, 84.24) --
	(355.85, 84.24);

\path[draw=drawColor,line width= 0.3pt,line join=round] ( 38.67,151.77) --
	(355.85,151.77);

\path[draw=drawColor,line width= 0.3pt,line join=round] ( 38.67,219.30) --
	(355.85,219.30);

\path[draw=drawColor,line width= 0.3pt,line join=round] ( 38.67,286.83) --
	(355.85,286.83);

\path[draw=drawColor,line width= 0.3pt,line join=round] ( 86.58, 36.08) --
	( 86.58,353.25);

\path[draw=drawColor,line width= 0.3pt,line join=round] (159.40, 36.08) --
	(159.40,353.25);

\path[draw=drawColor,line width= 0.3pt,line join=round] (232.21, 36.08) --
	(232.21,353.25);

\path[draw=drawColor,line width= 0.3pt,line join=round] (305.03, 36.08) --
	(305.03,353.25);

\path[draw=drawColor,line width= 0.6pt,line join=round] ( 38.67, 50.48) --
	(355.85, 50.48);

\path[draw=drawColor,line width= 0.6pt,line join=round] ( 38.67,118.00) --
	(355.85,118.00);

\path[draw=drawColor,line width= 0.6pt,line join=round] ( 38.67,185.53) --
	(355.85,185.53);

\path[draw=drawColor,line width= 0.6pt,line join=round] ( 38.67,253.06) --
	(355.85,253.06);

\path[draw=drawColor,line width= 0.6pt,line join=round] ( 38.67,320.59) --
	(355.85,320.59);

\path[draw=drawColor,line width= 0.6pt,line join=round] ( 50.18, 36.08) --
	( 50.18,353.25);

\path[draw=drawColor,line width= 0.6pt,line join=round] (122.99, 36.08) --
	(122.99,353.25);

\path[draw=drawColor,line width= 0.6pt,line join=round] (195.80, 36.08) --
	(195.80,353.25);

\path[draw=drawColor,line width= 0.6pt,line join=round] (268.62, 36.08) --
	(268.62,353.25);

\path[draw=drawColor,line width= 0.6pt,line join=round] (341.43, 36.08) --
	(341.43,353.25);
\definecolor{fillColor}{RGB}{51,51,51}

\path[fill=fillColor,fill opacity=0.30] ( 53.09, 50.51) --
	( 56.00, 50.53) --
	( 58.91, 50.55) --
	( 61.83, 50.59) --
	( 64.74, 50.59) --
	( 67.65, 50.62) --
	( 70.56, 50.62) --
	( 73.48, 50.70) --
	( 76.39, 50.79) --
	( 79.30, 50.88) --
	( 82.21, 50.74) --
	( 85.13, 50.80) --
	( 88.04, 50.95) --
	( 90.95, 51.20) --
	( 93.86, 50.97) --
	( 96.78, 51.03) --
	( 99.69, 51.12) --
	(102.60, 51.41) --
	(105.51, 51.48) --
	(108.43, 51.58) --
	(111.34, 51.81) --
	(114.25, 51.82) --
	(117.17, 52.22) --
	(120.08, 51.79) --
	(122.99, 52.19) --
	(125.90, 52.46) --
	(128.82, 54.00) --
	(131.73, 52.73) --
	(134.64, 52.75) --
	(137.55, 53.52) --
	(140.47, 54.34) --
	(143.38, 53.60) --
	(146.29, 55.73) --
	(149.20, 54.70) --
	(152.12, 55.16) --
	(155.03, 57.17) --
	(157.94, 56.84) --
	(160.85, 57.16) --
	(163.77, 57.85) --
	(166.68, 58.74) --
	(169.59, 57.77) --
	(172.50, 59.63) --
	(175.42, 59.98) --
	(178.33, 63.59) --
	(181.24, 59.85) --
	(184.15, 62.17) --
	(187.07, 71.59) --
	(189.98, 63.94) --
	(192.89, 64.23) --
	(195.80, 65.66) --
	(198.72, 71.01) --
	(201.63, 70.53) --
	(204.54, 71.81) --
	(207.45, 76.41) --
	(210.37, 71.90) --
	(213.28, 86.01) --
	(216.19, 75.20) --
	(219.11, 83.01) --
	(222.02, 79.27) --
	(224.93, 81.96) --
	(227.84, 91.20) --
	(230.76, 85.47) --
	(233.67,103.04) --
	(236.58, 97.12) --
	(239.49,109.86) --
	(242.41,105.80) --
	(245.32,105.85) --
	(248.23,109.45) --
	(251.14,115.60) --
	(254.06,117.27) --
	(256.97,127.09) --
	(259.88,122.25) --
	(262.79,125.02) --
	(265.71,134.92) --
	(268.62,154.22) --
	(271.53,146.19) --
	(274.44,151.23) --
	(277.36,147.08) --
	(280.27,156.12) --
	(283.18,168.98) --
	(286.09,166.51) --
	(289.01,179.14) --
	(291.92,185.12) --
	(294.83,205.71) --
	(297.74,183.28) --
	(300.66,194.65) --
	(303.57,200.88) --
	(306.48,229.91) --
	(309.39,231.71) --
	(312.31,246.41) --
	(315.22,254.06) --
	(318.13,265.99) --
	(321.04,267.54) --
	(323.96,278.83) --
	(326.87,287.15) --
	(329.78,301.86) --
	(332.70,307.37) --
	(335.61,316.09) --
	(338.52,335.74) --
	(341.43,338.84) --
	(341.43,308.32) --
	(338.52,303.06) --
	(335.61,292.70) --
	(332.70,276.94) --
	(329.78,268.50) --
	(326.87,266.89) --
	(323.96,257.99) --
	(321.04,254.21) --
	(318.13,244.49) --
	(315.22,224.68) --
	(312.31,221.84) --
	(309.39,211.39) --
	(306.48,207.40) --
	(303.57,192.10) --
	(300.66,187.23) --
	(297.74,179.46) --
	(294.83,174.25) --
	(291.92,171.60) --
	(289.01,167.26) --
	(286.09,151.77) --
	(283.18,148.20) --
	(280.27,144.63) --
	(277.36,141.42) --
	(274.44,136.67) --
	(271.53,127.99) --
	(268.62,125.01) --
	(265.71,124.95) --
	(262.79,122.09) --
	(259.88,117.92) --
	(256.97,108.15) --
	(254.06,105.55) --
	(251.14,104.29) --
	(248.23, 99.61) --
	(245.32, 97.40) --
	(242.41, 91.90) --
	(239.49, 94.97) --
	(236.58, 88.64) --
	(233.67, 93.91) --
	(230.76, 81.87) --
	(227.84, 84.35) --
	(224.93, 78.64) --
	(222.02, 76.60) --
	(219.11, 80.52) --
	(216.19, 72.31) --
	(213.28, 71.25) --
	(210.37, 70.24) --
	(207.45, 72.62) --
	(204.54, 67.08) --
	(201.63, 65.50) --
	(198.72, 68.61) --
	(195.80, 64.79) --
	(192.89, 62.55) --
	(189.98, 62.57) --
	(187.07, 64.06) --
	(184.15, 61.52) --
	(181.24, 59.08) --
	(178.33, 61.36) --
	(175.42, 59.48) --
	(172.50, 58.90) --
	(169.59, 56.84) --
	(166.68, 57.65) --
	(163.77, 57.00) --
	(160.85, 55.59) --
	(157.94, 55.40) --
	(155.03, 56.15) --
	(152.12, 54.72) --
	(149.20, 54.46) --
	(146.29, 54.57) --
	(143.38, 53.52) --
	(140.47, 53.43) --
	(137.55, 53.21) --
	(134.64, 52.50) --
	(131.73, 52.44) --
	(128.82, 52.65) --
	(125.90, 52.24) --
	(122.99, 52.00) --
	(120.08, 51.69) --
	(117.17, 51.90) --
	(114.25, 51.60) --
	(111.34, 51.72) --
	(108.43, 51.36) --
	(105.51, 51.32) --
	(102.60, 51.31) --
	( 99.69, 51.01) --
	( 96.78, 50.98) --
	( 93.86, 50.89) --
	( 90.95, 50.92) --
	( 88.04, 50.77) --
	( 85.13, 50.74) --
	( 82.21, 50.68) --
	( 79.30, 50.78) --
	( 76.39, 50.72) --
	( 73.48, 50.63) --
	( 70.56, 50.60) --
	( 67.65, 50.59) --
	( 64.74, 50.55) --
	( 61.83, 50.56) --
	( 58.91, 50.54) --
	( 56.00, 50.51) --
	( 53.09, 50.49) --
	cycle;
\definecolor{drawColor}{gray}{0.20}

\path[draw=drawColor,line width= 0.6pt,line join=round] ( 53.09, 50.49) --
	( 56.00, 50.51) --
	( 58.91, 50.54) --
	( 61.83, 50.57) --
	( 64.74, 50.56) --
	( 67.65, 50.59) --
	( 70.56, 50.60) --
	( 73.48, 50.65) --
	( 76.39, 50.75) --
	( 79.30, 50.84) --
	( 82.21, 50.70) --
	( 85.13, 50.77) --
	( 88.04, 50.78) --
	( 90.95, 51.01) --
	( 93.86, 50.93) --
	( 96.78, 50.99) --
	( 99.69, 51.06) --
	(102.60, 51.36) --
	(105.51, 51.40) --
	(108.43, 51.48) --
	(111.34, 51.74) --
	(114.25, 51.71) --
	(117.17, 52.06) --
	(120.08, 51.73) --
	(122.99, 52.05) --
	(125.90, 52.38) --
	(128.82, 53.04) --
	(131.73, 52.63) --
	(134.64, 52.65) --
	(137.55, 53.36) --
	(140.47, 53.49) --
	(143.38, 53.56) --
	(146.29, 54.78) --
	(149.20, 54.59) --
	(152.12, 54.98) --
	(155.03, 56.63) --
	(157.94, 56.14) --
	(160.85, 56.59) --
	(163.77, 57.19) --
	(166.68, 58.09) --
	(169.59, 57.14) --
	(172.50, 59.26) --
	(175.42, 59.63) --
	(178.33, 61.83) --
	(181.24, 59.28) --
	(184.15, 61.80) --
	(187.07, 66.11) --
	(189.98, 63.35) --
	(192.89, 63.54) --
	(195.80, 65.05) --
	(198.72, 69.49) --
	(201.63, 68.80) --
	(204.54, 68.01) --
	(207.45, 74.77) --
	(210.37, 71.42) --
	(213.28, 80.08) --
	(216.19, 73.38) --
	(219.11, 82.40) --
	(222.02, 77.51) --
	(224.93, 79.83) --
	(227.84, 86.95) --
	(230.76, 82.73) --
	(233.67, 98.81) --
	(236.58, 94.54) --
	(239.49,101.35) --
	(242.41, 95.40) --
	(245.32,102.30) --
	(248.23,103.26) --
	(251.14,109.40) --
	(254.06,109.96) --
	(256.97,111.76) --
	(259.88,119.31) --
	(262.79,124.09) --
	(265.71,132.49) --
	(268.62,138.77) --
	(271.53,138.89) --
	(274.44,143.06) --
	(277.36,144.19) --
	(280.27,153.57) --
	(283.18,156.86) --
	(286.09,163.73) --
	(289.01,171.27) --
	(291.92,173.18) --
	(294.83,186.66) --
	(297.74,181.72) --
	(300.66,191.04) --
	(303.57,197.77) --
	(306.48,219.83) --
	(309.39,221.54) --
	(312.31,226.60) --
	(315.22,230.05) --
	(318.13,251.81) --
	(321.04,256.14) --
	(323.96,262.24) --
	(326.87,271.31) --
	(329.78,276.01) --
	(332.70,285.82) --
	(335.61,303.10) --
	(338.52,309.17) --
	(341.43,314.59);
\end{scope}
\begin{scope}
\definecolor{drawColor}{RGB}{190,190,190}

\path[draw=drawColor,line width= 0.6pt,line join=round] ( 38.67, 36.08) --
	( 38.67,353.25);
\end{scope}
\begin{scope}
\definecolor{drawColor}{gray}{0.30}

\node[text=drawColor,anchor=base east,inner sep=0pt, outer sep=0pt, scale=  0.88] at ( 33.72, 47.45) {0};

\node[text=drawColor,anchor=base east,inner sep=0pt, outer sep=0pt, scale=  0.88] at ( 33.72,114.97) {50};

\node[text=drawColor,anchor=base east,inner sep=0pt, outer sep=0pt, scale=  0.88] at ( 33.72,182.50) {100};

\node[text=drawColor,anchor=base east,inner sep=0pt, outer sep=0pt, scale=  0.88] at ( 33.72,250.03) {150};

\node[text=drawColor,anchor=base east,inner sep=0pt, outer sep=0pt, scale=  0.88] at ( 33.72,317.56) {200};
\end{scope}
\begin{scope}
\definecolor{drawColor}{gray}{0.20}

\path[draw=drawColor,line width= 0.6pt,line join=round] ( 35.92, 50.48) --
	( 38.67, 50.48);

\path[draw=drawColor,line width= 0.6pt,line join=round] ( 35.92,118.00) --
	( 38.67,118.00);

\path[draw=drawColor,line width= 0.6pt,line join=round] ( 35.92,185.53) --
	( 38.67,185.53);

\path[draw=drawColor,line width= 0.6pt,line join=round] ( 35.92,253.06) --
	( 38.67,253.06);

\path[draw=drawColor,line width= 0.6pt,line join=round] ( 35.92,320.59) --
	( 38.67,320.59);
\end{scope}
\begin{scope}
\definecolor{drawColor}{RGB}{190,190,190}

\path[draw=drawColor,line width= 0.6pt,line join=round] ( 38.67, 36.08) --
	(355.85, 36.08);
\end{scope}
\begin{scope}
\definecolor{drawColor}{gray}{0.20}

\path[draw=drawColor,line width= 0.6pt,line join=round] ( 50.18, 33.33) --
	( 50.18, 36.08);

\path[draw=drawColor,line width= 0.6pt,line join=round] (122.99, 33.33) --
	(122.99, 36.08);

\path[draw=drawColor,line width= 0.6pt,line join=round] (195.80, 33.33) --
	(195.80, 36.08);

\path[draw=drawColor,line width= 0.6pt,line join=round] (268.62, 33.33) --
	(268.62, 36.08);

\path[draw=drawColor,line width= 0.6pt,line join=round] (341.43, 33.33) --
	(341.43, 36.08);
\end{scope}
\begin{scope}
\definecolor{drawColor}{gray}{0.30}

\node[text=drawColor,anchor=base,inner sep=0pt, outer sep=0pt, scale=  0.88] at ( 50.18, 25.06) {0};

\node[text=drawColor,anchor=base,inner sep=0pt, outer sep=0pt, scale=  0.88] at (122.99, 25.06) {25};

\node[text=drawColor,anchor=base,inner sep=0pt, outer sep=0pt, scale=  0.88] at (195.80, 25.06) {50};

\node[text=drawColor,anchor=base,inner sep=0pt, outer sep=0pt, scale=  0.88] at (268.62, 25.06) {75};

\node[text=drawColor,anchor=base,inner sep=0pt, outer sep=0pt, scale=  0.88] at (341.43, 25.06) {100};
\end{scope}
\begin{scope}
\definecolor{drawColor}{RGB}{0,0,0}

\node[text=drawColor,anchor=base,inner sep=0pt, outer sep=0pt, scale=  1.50] at (197.26,  7.71) {$m$};
\end{scope}
\begin{scope}
\definecolor{drawColor}{RGB}{0,0,0}

\node[text=drawColor,rotate= 90.00,anchor=base,inner sep=0pt, outer sep=0pt, scale=  1.50] at ( 17.19,194.66) {Execution time in ms};
\end{scope}
\end{tikzpicture}

%% file: lambda_power.tex
\begin{table}[ht]
\centering
\begin{tabular}{lrrrrrr}
	\hline
	& Eff Sz. $\theta$ & $\mathbb{E}(M_m)$ & n & est. $\lambda_{90}$-pwr & $\lambda_{90}$-pwr & Diff in std \\ 
	\hline
	1 & 0.60000 &      5 &     70 & 0.24900 & 0.26691 & 1.23987 \\ 
	2 & 0.60000 &      5 &     80 & 0.39600 & 0.39977 & 0.24081 \\ 
	3 & 0.60000 &      5 &     90 & 0.53800 & 0.53479 & 0.18527 \\ 
	4 & 0.60000 &      5 &    100 & 0.65700 & 0.65626 & 0.05057 \\ 
	5 & 0.60000 &     20 &     50 & 0.02800 & 0.02538 & 0.48379 \\ 
	6 & 0.60000 &     20 &     60 & 0.15700 & 0.13890 & 1.69096 \\ 
	7 & 0.60000 &     20 &     70 & 0.37800 & 0.36864 & 0.61103 \\ 
	8 & 0.60000 &     20 &     80 & 0.59900 & 0.62231 & 1.50514 \\ 
	9 & 0.60000 &     60 &     40 & 0.00200 & 0.00143 & 0.43439 \\ 
	10 & 0.60000 &     60 &     50 & 0.09900 & 0.08584 & 1.49443 \\ 
	11 & 0.60000 &     60 &     60 & 0.49200 & 0.49307 & 0.06522 \\ 
	12 & 0.60000 &    100 &     30 & 0.00000 & 0.00000 &    Inf \\ 
	13 & 0.60000 &    100 &     40 & 0.00600 & 0.00658 & 0.22041 \\ 
	14 & 0.60000 &    100 &     50 & 0.27000 & 0.30726 & 2.80406 \\ 
	15 & 0.80000 &      5 &     40 & 0.25200 & 0.26037 & 0.59552 \\ 
	16 & 0.80000 &      5 &     50 & 0.50200 & 0.49870 & 0.19588 \\ 
	17 & 0.80000 &      5 &     60 & 0.70400 & 0.70951 & 0.39928 \\ 
	18 & 0.80000 &     20 &     30 & 0.03600 & 0.03969 & 0.60858 \\ 
	19 & 0.80000 &     20 &     40 & 0.35700 & 0.36732 & 0.64187 \\ 
	20 & 0.80000 &     60 &     20 & 0.00000 & 0.00004 & 0.18826 \\ 
	21 & 0.80000 &     60 &     30 & 0.14700 & 0.15775 & 0.88238 \\ 
	22 & 0.80000 &    100 &     20 & 0.00300 & 0.00013 & 8.59402 \\ 
	23 & 0.80000 &    100 &     30 & 0.50600 & 0.48563 & 1.42549 \\ 
	24 & 1.00000 &      5 &     30 & 0.39200 & 0.39421 & 0.14658 \\ 
	25 & 1.00000 &     20 &     20 & 0.04500 & 0.04534 & 0.05420 \\ 
	26 & 1.00000 &     20 &     30 & 0.61400 & 0.63660 & 1.40588 \\ 
	27 & 1.00000 &     60 &     20 & 0.22500 & 0.19941 & 2.00156 \\ 
	28 & 1.00000 &    100 &     20 & 0.58600 & 0.57569 & 0.64138 \\ 
	\hline
\end{tabular}
\caption{Faithfully rounded calculation of the $\lambda_{90}-$power of the Benjamini Hochberg procedure (controlling the FDR at $\alpha=0.15$) when applied to $m=200$ test-statistics with chi-squared distributions $F_{\nu,0}$ under the null hypothesis and $F_{\nu,\mu}$ under the alternative where $\nu=2n-2$, $\mu=\sqrt{\frac{n}{2}}\theta$. The results are under the RM-model where $\pi_0=\frac{\mathbb{E}(M_m)}{m}$. For comparison the fourth column contains the Monte Carlo estimates  (sample size $1000$) given in \cite[Table 3]{izmirlian2018average}. The following column gives the faithfully rounded power calculated using our method and the last column states the absolute difference between the previous two columns dived by the standard deviation of the Monte Carlo estimation (which was estimated using 300 replicates). As expected most Monte Carlo approximations are within one or two standard deviations of the faithfully rounded result. For the twelfth row no value is given since the estimated standard deviation was zero. This is not unexpected since the faithfully rounded result is $\approx4.58\cdot10^{-7}$ which is two orders magnitude smaller than $(300\cdot 1000)^{-1}$.}
\label{table:lambda_power}
\end{table}

%% file: 7-discussion.tex
\section{Discussion} \label{sec7}
We have presented computationally efficient and numerically stable methods for calculating the joint distribution of order statistics. Such joint distributions have a multitude of important applications that require their repeated evaluation (to numerically solve optimization problems). Apart from the applications that we have presented in Section \ref{sec6}, they include, among others, the calibration goodness-of-fit tests with equal local levels (see Section 1.4 of \cite{gontscharuk2016}) and the adjustment of the asymptotically optimal rejection curve as proposed by 
\cite{FiGoDi2012}, see Equation (19) in their paper, and \cite[Equation (6.1)]{AORC} with the goal of obtaining valid critical values for a step-up-down procedure (guaranteeing strict FDR control). The latter applications have not been considered explicitly in the present work, because they merely refer to the one-group case. For this case, the methods of \cite{moscovich2016} are already sufficiently accurate and fast.

Future extensions to our methods could include a normalization of the exponents (in  
Noe's recursion) to avoid underflows and the exploration of potential efficiency gains in the exact computation of Bolshev's recursion by a trade-off between the memory consumption and the frequency of normalizations of the intermediate rational numbers.

A preliminary version of our planned package (which utilizes RCPP, cf. \cite{RcppBook}) for the R language (\cite{Rcitation}) is available at  \url{https://github.com/jvschroeder/OrdStat/} and can be installed using the \href{https://cran.r-project.org/web/packages/devtools/index.html}{devtools package}:
\begin{verbatim}
	install.pacakges("devtools")
	devtools::install_github("jvschroeder/OrdStat")
\end{verbatim}

The code used to generate the graphics and numerical examples is available at \url{https://github.com/jvschroeder/OrdStatExamples/}. The graphics were created using ggplot2 (\cite{ggplot2}) and the R package \href{https://cran.r-project.org/web/packages/tikzDevice/index.html}{tikzDevice}.

%% file: a-proofs.tex
\section{Proofs}
\begin{proofof}{Lemma \ref{lemma:generalized_bolshev_recursion}}
    Let $c_X(t):=\left|\left\{U_i\leq t|i\in X\right\}\right|$, $m_1\in[n_1]$, $m_2\in[n_2]$ and $m:=m_1+m_2$. Then, mimicking the approach of \cite[p. 367 ff.]{shorackwellner} and \cite[Proposition 1]{blanchard2014least}, it holds that
    \begin{align*}
       &1-\Psi(m_1,m_2):=1-\mathbb{P}\left(U_{1:m}\leq b_1,\cdots, U_{m:m}\leq b_m\right)\\
        =&1-\mathbb{P}\left(\bigcap_{k=1}^m c_{[m]}(b_k)\geq k\right)\\
        =&\mathbb{P}\left(\exists k\in[m]:c_{[m]}(b_k)=k-1\right)\\
        =&\sum_{k=1}^m\mathbb{P}\left(\left[c_{[m]}(b_k)=k-1\right]\cap\left[\bigcap_{j=1}^{k-1} c_{[m]}(b_j)\geq j\right]\right)\\
        =&1-\sum_{k=0}^{m-1}\mathbb{P}\left(\left[c_{[m]}(b_{k+1})=k\right]\cap\left[\bigcap_{j=1}^{k} c_{[m]}(b_j)\geq j\right]\right)\\
        =&\sum_{k=0}^{m-1}\sum_{\substack{X\subset [m]\\|X|=k}}\mathbb{P}\left(\left[\bigcap_{i\in[m]\setminus X} U_i>b_{k+1}\right]\cap\left[\bigcap_{j=1}^{k} c_X(b_j)\geq j\right]\right)\\
        =&\sum_{\substack{0\leq k_1\leq m_1\\0\leq k_2\leq m_2\\k_1+k_2<m}}\sum_{\substack{X\subset [m]\\|X|=k_1+k_2\\|X\cap[m_1]|=k_1}}\mathbb{P}\left(\bigcap_{i\in[m]\setminus X} U_i>b_{k_1+k_2+1}\right)\cdot\mathbb{P}\left(\bigcap_{j=1}^{k_1+k_2} c_X(b_j)\geq j\right)\\
        =&\sum_{\substack{0\leq k_1\leq m_1\\0\leq k_2\leq m_2\\k_1+k_2<m}}\binom{m_1}{k_1}\binom{m_2}{k_2}(1-b_{k_1+k_2+1})^{m_1-k_1}\cdot(1-F(b_{k_1+k_2+1}))^{m_2-k_2}\\
        &\times \mathbb{P}\left(U_{1:(k_1+k_2)}\leq b_1,\cdots, U_{(k_1+k_2):(k_1+k_2)}\leq b_{k_1+k_2}\right)
    \end{align*} Since this holds for any $m_1\in[n_1]$, $m_2\in[n_2]$ it follows that\begin{equation*}
    	\Psi(m_1,m_2)=1-\sum_{\substack{0\leq k_1\leq m_1\\0\leq k_2\leq m_2\\k_1+k_2<m}}M^{(m_1,m_2)}_{k_1,k_2}\cdot\Psi(k_1,k_2)
    \end{equation*} where $M$ is given by \eqref{eq:m_general_bolshev}.
        
    The recursions for $M$ follow from the definition of the binomial coefficient and routine calculations.
\end{proofof}

\begin{proofof}{Lemma \ref{lemma:generalized_noe_rec}}
    Let  $(a_i)_{i\in[n]}$ be an increasing sequence in $[0,1]$ such that $\forall i\in[n]:a_i<b_i$. Using notation similar to that of \cite[p. 362 ff.]{shorackwellner} let $i:=i_1+i_2$ and \begin{align}
        Q_{i_1,i_2}(m)&:=\mathbb{P}\left(\bigcap_{j=1}^{i_1+i_2}\left[a_j<X_{j:i}\leq b_j\bigcap X_{j:i}\leq c_m\right]\right)\label{eq:noe_q_def}
    \end{align} where the $c_j$ are the $2i$ boundaries $a_1,\cdots,a_{i},b_1,\cdots,b_{i}$ arranged in any ascending order. To extend the induction in \cite[p. 364 ff.]{shorackwellner} to the two-group case we only need to mimick the the approach of \cite[Proposition 1]{blanchard2014least} to provide a recursive formula for $Q_{i_1,i_2}(m)$, given that $h(m+1)-1\leq i\leq g(m-1)$ (where $g,h$ are given by \cite[p. 362, Eq. (17) and Eq. (18)]{shorackwellner}). To this end note that in this case\begin{align*}
        Q_{i_1,i_2}(m)&=\sum_{k=0}^{i}\sum_{\substack{M\subset[i]\\|M|=k}} \mathbb{P}\left(\bigcap_{j=1}^k\left[a_j< X_{j:M}\leq b_j\bigcap X_{j:M}\leq c_{m-1}\right]\right.\\
        &\cap\left.\bigcap_{j=1}^{i-k}\left[a_{k+j}<X_{j:\bar M}\leq b_{k+j}\bigcap c_{m-1}<X_{j:\bar M}\leq c_m\right]\right)\\
        &=\sum_{k_1=0}^{i_1}\sum_{k_2=0}^{ i_2}\sum_{\substack{M\subset [i]\\|M|=k_1+k_2\\|M\cap[i_1]|=k_1}}\mathbb{P}\left(\bigcap_{j=1}^k\left[a_j< X_{j:M}\leq b_j\bigcap X_{j:M}\leq c_{m-1}\right]\right)\\
        &\times\mathbb{P}\left(\bigcap_{j=1}^{i-k}\left[a_{k+j}<X_{j:\bar M}\leq b_{k+j}\bigcap c_{m-1}<X_{j:\bar M}\leq c_m\right]\right)\\
        &=\sum_{k_1=0}^{i_1}\sum_{k_2=0}^{ i_2}\sum_{\substack{M\subset [i]\\|M|=k_1+k_2\\|M\cap[i_1]|=k_1}}Q_{k_1,k_2}(m-1)\\
        &\times\mathbb{P}\left(\bigcap_{j=1}^{i-k}\left[a_{k+j}<X_{j:\bar M}\leq b_{k+j}\bigcap c_{m-1}<X_{j:\bar M}\leq c_m\right]\right)\\
        &=\sum_{k_1=0}^{i_1}\sum_{k_2=0}^{ i_2}\sum_{\substack{M\subset [i]\\|M|=k_1+k_2\\|M\cap[i_1]|=k_1}}Q_{k_1,k_2}(m-1)\times\mathbb{P}\left(\bigcap_{j=1}^{i-k}c_{m-1}<X_{j:\bar M}\leq c_m\right)\\
&=\sum_{k_1=0}^{i_1}\sum_{k_2=0}^{ i_2}\binom{i_1}{k_1}\binom{i_2}{k_2}Q_{k_1,k_2}(m-1)\times(c_m-c_{m-1})^{i_1-k_1}\\
&\times (F(c_m)-F(c_{m-1}))^{i_2-k_2}\\
        &=\sum_{\substack{0\leq k_1\leq i_1\\0\leq k_2\leq i_2\\ h(m)-1\leq k_1+k_2}}\binom{i_1}{k_1}\binom{i_2}{k_2}Q_{k_1,k_2}(m-1)\times(c_m-c_{m-1})^{i_1-k_1}\\
        &\times (F(c_m)-F(c_{m-1}))^{i_2-k_2}
    \end{align*} and $\bar M:=[i]\setminus M$ denotes the complement of $M$.

    If all $a_i=0$, then it holds that \begin{equation*}
    	g(i)=\begin{cases}
			i&i<n\\
			n&i\geq n
			\end{cases}\quad\mathrm{and}\quad
		h(i)=\begin{cases}
			1&i<n\\
			i-n&i\geq n
		\end{cases}\quad\mathrm{and}\quad
        c_i=\begin{cases}
            0&i\leq n\\
            b_{i-n}&n<i\leq 2n\\
            1&i=2n+1
        \end{cases}\\
    \end{equation*} which implies 
    \begin{equation*}
        F(c_m)-F(c_{m-1})=\begin{cases}
            0&m\leq n\\
            F(b_1)&m=n+1\\
            F(b_{m-n})-F(a_{m-n-1})&n+1<m\leq2n\\
            1-F(b_n)&m=2n+1
        \end{cases}.
    \end{equation*} From the \eqref{eq:noe_q_def} it follows that $Q_i(n)=\begin{cases}
    1&i=0\\
    0&i>0
    \end{cases}$. Thus ${Q_{i_1,i_2}\left(n+1\right)=b_1^{i_1}\cdot F(b_1)^{i_2}}$ and for $m>n+1$ it holds that\begin{multline*}
        Q_{i_1,i_2}(m)=\sum_{\substack{0\leq k_1\leq i_1\\0\leq k_2\leq i_2\\ m-n-1\leq k_1+k_2}}\binom{i_1}{k_1}\binom{i_2}{k_2}Q_{k_1,k_2}(m-1)\\
        \times(c_m-c_{m-1})^{i_1-k_1}\times (F(c_m)-F(c_{m-1}))^{i_2-k_2}
    \end{multline*} which needs to be calculated for $m-n\leq i\leq n$.

    Let $b_0:=0$, $\tilde Q_{0,0}(0):=1$, $\tilde Q_{i_1,i_2}(1):=b_1^{i_1}\cdot F(b_1)^{i_2}$ and for $m>1$\begin{multline*}
    \tilde Q_{i_1,i_2}(m):=\sum_{\substack{0\leq k_1\leq i_1\\0\leq k_2\leq i_2\\ m-1\leq k_1+k_2}}\binom{i_1}{k_1}\binom{i_2}{k_2}Q_{k_1,k_2}(m-1)\times(b_m-b_{m-1})^{i_1-k_1}\\
    \times (F(b_{m})-F(b_{m-1}))^{i_2-k_2}
\end{multline*} (which needs to be calculated for $\forall(i_1,i_2)\in[n_1]\times[n_2]:m\leq i_1+i_2\leq n$). Then it holds that
\begin{equation*}
    \Psi(i_1,i_2)=\tilde Q_{i_1,i_2}(i_1+i_2)
\end{equation*} for all $i_1\in[n_1],i_2\in[n_2]$.
\end{proofof}

\begin{proofof}{Lemma \ref{lemma:rec_complexity}}
    Counting the numer of operations in the loops of Algorithm \ref{alg:bolshev_general} it follows that
    \begin{align*}
    \text{\#Operations}&<\sum_{m_1=0}^{n_1}\left[7+\sum_{m_2=0}^{n_2}\sum_{k_1=0}^{m_1}\left[6+\sum_{k_2=0}^{m_2}10\right]\right]=\cdots\\
    &=\frac{5n_1^2n_2^2+21n_1^2n_2+15n_1n_2^2+63n_1n_2}{2}\\
    &+8n_1^2+5n_2^2+31n_1+21n_2+23
    \end{align*} holds. For the space complexity simply note that, to use the recursions for $M^{(m_1,m_2)}_{k_1,k_2}$, we need to keep track of at most $M^{(m_1,m_2-1)}_{k_1,k_2}$ and $M^{(m_1-1,m_2)}_{k_1,k_2}$ (which is pessimistic - cf. algorithm \ref{alg:bolshev_general}).
    
    For Steck's recursion first note that, using exponentiation by squaring, one can calculate $a^n$ in $O(\log_2(n))$ multiplications (where $n\in\mathbb{N}$, c.f. \cite[p. 442, Algorithm A]{KnuthTAOCP02}). Thus the first row and last column of $M^{(m_1,m_2)}$ (cf. equations \eqref{eq:steck_coeff1} \eqref{eq:steck_coeff2}) can be calculated in\begin{align*}
        &\sum_{j=1}^{m_2-1}\left[O(\log_2(m_1))+O(\log_2(m_2-j))\right]+\sum_{j=1}^{m_1-1}\left[O(\log_2(m_1-j))\right]\\
        &=O(\log_2(m_2\cdot m_1))+\sum_{j=1}^{m_2}\left[O(\log_2(j))\right]+\sum_{j=1}^{m_1-1}\left[O(\log_2(j))\right]\\
        &\subset O(\log_2(m_2\cdot m_1)+m_1\log_2(m_1)+m_2\log_2(m_2))\\
        &=O(m_1\log_2(m_1)+m_2\log_2(m_2))
    \end{align*} Thus to calculate all the coefficient matrices we need at most\begin{align*}
        &\sum_{m_1=1}^{n_1}\sum_{m_2=1}^{n_2}O(m_1\log_2(m_1)+m_2\log_2(m_2))\\
        &\subset O(n_1n_2\log_2(n_1)+n_1n_2\log_2(n_2))
    \end{align*} arithmetic operations (since we can calculate $a(m_1,j_1),a(m_2,j_2)$ for $j_1\in[m_1],j_2\in[m_2]$ in $O(m_1+m_2)$ using \eqref{eq:steck_binomial_rec}). It remains to note that \eqref{eq:steck_rec_eq} needs at most $O(m_1m_2)$ arithmetic operations. For the space complexity simply note that we do not need to keep track of the previous coefficient matrices.
    
    For Noe's recursion first note that, for every $m\in\mathbb{N}$, we can calculate $a^{(m),1}(j_1),a^{(m),2}(j_2)$ for $j_1\in[i_1],j_2\in[i_2]$ in $O(j_1+j_2)$. Furthermore, using \eqref{eq:steck_binomial_rec}, the binomial coefficients in \eqref{eq:noe_coeff_rec} can be calculated in $O(i_1+i_2)$. Thus $M^{i_1,i_2}$ can be calculated in $O(i_1i_2)$. Thus $Q_{i_1,i_2}(m)$ (assuming the necessary $Q_{\cdot,\cdot}(m-1)$ have already been calculated) is $O(i_1i_2)$. Therefore the overall computational complexity is at most\begin{align*}
        \sum_{i_1=1}^{n_1}\sum_{i_2=1}^{n_2}O(i_1i_2)&=O(i_1^2i_2^2(i_1+i_2))
    \end{align*} For the space complexity simply note, again, that we do not need to keep track of the previous coefficient matrices.
\end{proofof}

\begin{proofof}{Lemma \ref{lemma:inexact_thresholds}}
    By Noe's recursion (cf. lemma \ref{lemma:generalized_noe_rec}) the probability $\Psi(i_1,i_2)$ can be obtained by evaluating a polynomial of degree $i_1+i_2$ with only positive coefficients at $\mathbf{x}\in\mathbb{R}^{2n}$ (where $\mathbf{x}$ is given by \eqref{eq:inexact_thresholds_x_def}). It is therefore sufficient to show that the statement is true for such polynomials when applied to non-negative arguments. To provide a concise proof we utilise interval arithmetic (cf. \cite{k96}). Due to linearity (since all coefficients and inputs are non-negative) it is sufficient to show the claim for monomials $p(x_1,\cdots,x_{2n}):=\prod_{i=1}^{2n} x_i^{a_i}$ with $a_i\in\mathbb{N}$, $\sum_{i=1}^{2n} a_i\leq i_1+i_2$. Since $0 <1-\varepsilon<1$ and $1+\varepsilon>1$ it follows that \begin{equation*}
        \tilde x_i^{a_i}\in x_i^{a_i}\cdot\left(\left(1-2\varepsilon\right)^{a_i},\left(1+2\varepsilon\right)^{a_i}\right)
    \end{equation*} which implies\begin{equation*}
    	p(\tilde x_1,\cdots,\tilde x_{2n})\in p( x_1,\cdots, x_{2n})\cdot\left((1-2\varepsilon)^{i_1+i_2},(1+2\varepsilon)^{i_1+i_2}\right)
    \end{equation*} due to \cite[Equation (4)]{k96}.
\end{proofof}

\begin{proofof}{Lemma \ref{lemma:exact_dist}}
By \cite[Theorem 3.1]{roquain2011} under the unconditional model $RM(m,\pi_0,F)$ and for a step-up procedure $\text{SU}_\mathbf{t}$:

\begin{align*}
    \mathbb{P}\left(V(\text{SU}_\mathbf{t},p)=j,R(\text{SU}_\mathbf{t},\mathbf{p})=k\right)&=\mathbb{P}\left(V(\text{SU}_\mathbf{t},\mathbf{p})=j\left|R(\text{SU}_\mathbf{t},\mathbf{p})=k\right.\right)\cdot \mathbb{P}\left(R(\text{SU}_\mathbf{t},\mathbf{p})=k\right)\\
    &=\binom{m}{k}\binom{k}{j}\tilde\pi_0^j(1-\tilde\pi_0)^{k-j}G(t_k)^k\tilde\Psi_{m-k}(t_m,\cdots,t_{k+1})
\end{align*} where\begin{align*}
    \tilde\Psi_{m-k}(t_m,\cdots,t_{k+1})&:=\Psi_{m-k,0}^{\text{Uni}[0,1],F}(1-G(t_m),\cdots,1-G(t_{k+1}))\\
    \tilde\pi_0&:=\frac{\pi_0t_k}{G(t_k)}\\
    G(t)&:=\pi_0t+(1-\pi_0)F(t)
\end{align*} and $\mathbb{P}$ denotes $\mathbb{P}_{m,\pi_0,F}$.

Furthermore under the conditional model $FM(m,m_0,F)$ and for a step-up procedure $\text{SU}_\mathbf{t}$ it holds (by \cite[Section 5.3]{roquain2011}) that (where $\bar F(t):=1-F(1-t)$):\begin{align*}
   &\mathbb{P}\left(V(\text{SU}_\mathbf{t},\mathbf{p})=j,R(\text{SU}_\mathbf{t},\mathbf{p})=k\right)\\
   &=\binom{m_0}{j}\binom{m-m_0}{k-j}t_k^j\left(F(t_k)\right)^{k-j}\Psi^{\text{Uni}[0,1],\bar F}_{\substack{m-k-(m_0-j),\\ m_0-j}}(1-t_m,\cdots,1-t_{k+1})
\end{align*} where $\mathbb{P}$ denotes $\mathbb{P}_{m,m_0,F}$.
\end{proofof}

\newpage
\section{Algorithms}
\begin{algorithm}
\caption{Bolshev Recursion}\label{alg:bolshev}
\begin{algorithmic}[1]
    \Procedure{Bolshev}{$b$}
        \State $b\gets 1-b$
        \State $s\in\mathbb{R}^{n}$
        \State $s_1\gets 0$
        \For{$k=2,\cdots,n$}
            \State $v\gets 1$
            \For{$j=1,\cdots,k-1$}
                \State $v\gets v-s_{j}$
                \State $s_{j}\gets \frac{s_{j} \cdot b_{j} \cdot k}{k-(j-1)}$
            \EndFor
            \State $s_k\gets k\cdot v\cdot b_{k}$
        \EndFor
        \State\Return $1-\sum_{i=1}^n s_i$
    \EndProcedure
\end{algorithmic}
\end{algorithm}

\begin{algorithm}
\caption{Efficient Generalized Bolshev Recursion}\label{alg:bolshev_general}
\begin{algorithmic}[1]
    \Procedure{GeneralizedBolshev}{$n_1\in\mathbb{N}$,$n_2\in\mathbb{N}$, $v^{(1)}\in(0,1)^{n_1+n_2}$,$v^{(2)}\in(0,1)^{n_1+n_2}$}
        \State $r\in\mathbb{R}^{(n_1+1)\times (n_2+1)}$
        \State $\forall i\in[n_1+1],j\in[n_2+1]:r_{i,j}\gets 1$
        \State $M\gets r$
        \State $M^{(0)}\in\mathbb{R}^{n_1+1}$
        \State $\forall i\in[n_1+1]:M^{(0)}_{i}\gets 1$
        \For{$m_1=0,\cdots,n_1$}
            \For{$m_2=0,\cdots,n_2$}
                \For{$k_1=0,\cdots,m_1$}
                    \For{$k_2=0,\cdots,m_2$}
                        \If{$k_1<m_1\lor k_2<m_2$}
                            \State $r_{m_1+1,m_2+1}\gets r_{m_1+1,m_2+1} - M_{k_1+1,k_2+1}\cdot r_{k_1+1,k_2+1}$
                        \EndIf
                        \If{$m_2<n_2$}
                            \State $M_{k_1+1,k_2+1}\gets M_{k_1+1,k_2+1}\cdot\frac{m_2+1}{m_2+1-k_2}\cdot\left(1-v^{(2)}_{k_1+k_2+1}\right)$
                        \EndIf
                    \EndFor
                    \If{$m_2<n_2\land k_1<m_1$}
                        \State $M_{k_1+1,m_2+2}\gets M_{k_1+2,m_2+1}\cdot \frac{k_1+1}{m_1-k_1}\cdot\left(1-v^{(1)}_{k_1+m_2+2}\right)$
                    \EndIf
                \EndFor
            \EndFor
            \If{$m_1<n_1$}
                \For{$k_1=0,\cdots,m_1$}
                    \State $M^{(0)}_{k_1+1}\gets M^{(0)}_{k_1+1}\cdot \frac{m_1+1}{m_1+1-k_1}\cdot\left(1-v^{(1)}_{k_1+1}\right)$
                    \State $M_{k_1+1,1}\gets M^{(0)}_{k_1+1}$
                \EndFor
            \EndIf
        \EndFor
        \State\Return $r_{n_1+1,n_2+1}$
    \EndProcedure
\end{algorithmic}
\end{algorithm}